\documentclass[10pt]{article} % For LaTeX2e
\usepackage[preprint]{tmlr}
\usepackage{amsmath,amsfonts,bm}

\def\eqref#1{Equation~\ref{#1}}
\def\Eqref#1{(Equation~\ref{#1})}
\def\1{\bm{1}}

\DeclareMathAlphabet{\mathsfit}{\encodingdefault}{\sfdefault}{m}{sl}
\SetMathAlphabet{\mathsfit}{bold}{\encodingdefault}{\sfdefault}{bx}{n}

\newcommand{\E}{\mathbb{E}}

\usepackage{hyperref}
\usepackage{url}
\usepackage[pdftex]{graphicx}
\usepackage[ruled, vlined]{algorithm2e}
\usepackage{caption}
\usepackage{subcaption}
\usepackage{wrapfig}
\usepackage{enumitem}
\usepackage{float}

\title{Fast approximate Bayesian multidimensional scaling with consistency guarantees}

\author{\name Ami Sheth \email amisheth26@g.ucla.edu \\
      \addr Department of Biostatistics\\
      University of California, Los Angeles 
      \AND
      \name Aaron Smith \email smith.aaron.matthew@gmail.com \\
      \addr Department of Mathematics and Statistics \\
      University of Ottawa
      \AND
      \name Andrew J.~Holbrook \email aholbroo@g.ucla.edu\\
      \addr Department of Biostatistics\\
      University of California, Los Angeles}

\def\month{05}  % Insert correct month for camera-ready version
\def\year{2026} % Insert correct year for camera-ready version
\def\openreview{\url{https://openreview.net/forum?id=XXXX}} % Insert correct link to OpenReview for camera-ready version

\newtheorem{theorem}{Theorem}
\newtheorem{lemma}{Lemma}
\newtheorem{proposition}{Proposition}

\newtheorem{assumption}{Assumption}
\newtheorem{remark}{Remark}
\newtheorem{definition}{Definition}
\newtheorem{proof}{Proof}

\newcommand{\RR}{\mathbb R}
\newcommand{\mb}{\mathbf}

\newcommand{\Pbb}{\mathbb{P}}

\newcommand{\cG}{\mathcal{G}}
\newcommand{\cO}{\mathcal{O}}
\newcommand{\cN}{\mathcal{N}}
\newcommand{\norm}[1]{\left\|#1\right\|}

\newcommand{\latentData}{\mb{X}}
\newcommand{\latentdata}{\mb{x}}
\newcommand{\highdata}{\mb{y}}

\newcommand{\distanceMatrix}{\mb{D}}

\begin{document}

\maketitle

\begin{abstract}
Bayesian multidimensional scaling (BMDS) embeds $n$ objects in a low-dimensional space to approximately preserve an observed dissimilarity matrix. Compared to classic MDS, BMDS is more robust to model misspecification and supports posterior uncertainty quantification and joint estimation within hierarchical models. However, standard BMDS inference is computationally prohibitive, requiring $O(n^2)$ operations per MCMC iteration to evaluate the likelihood. We propose Barnes--Hut BMDS (BH-BMDS), which uses a tree-based approximation to the likelihood and a Gibbs sampler that leverages this structure, remaining compatible with hierarchical extensions. BH-BMDS reduces computational complexity to $O(n \log n)$ while preserving the geometric fidelity of the embedding. We further establish consistency for the stationary measure of BH-BMDS, proving that it concentrates around the true latent configuration even as the total error of the surrogate likelihood diverges. Notably, this consistency holds in the infinite-dimensional limit. We evaluate the approximation on datasets with diverse structure, including air traffic networks, arXiv abstracts, MNIST images and neural activity recordings from mouse models of tau pathology. Across all settings, BH-BMDS closely matches BMDS while achieving substantial computational gains, with approximately 10-fold speedups at $n=1{,}000$ and 70-fold speedups at $n=10{,}000$. These gains increase with $n$, demonstrating strong empirical scalability.
\end{abstract}

\section{Introduction} \label{sec:intro}
Dissimilarity data arise in many domains, including psychology, physics and biology, where a goal is to quantify how pairs of complex or high-dimensional objects differ from one another. A common approach is to embed objects into a low-dimensional space, enabling visualization and downstream modeling while preserving the observed dissimilarities. Bayesian multidimensional scaling (BMDS) provides a probabilistic framework for this task, mapping $n$ objects to latent coordinates in a low-dimensional Euclidean space while modeling observed dissimilarities through a likelihood-based formulation. Unlike classical approaches such as PCA \citep{pca_pearson, pca_hotelling} and MDS \citep{torgerson1952cmds}, or more recent approaches such as t-SNE \citep{tsne2008} and UMAP \citep{umap}, BMDS enables posterior uncertainty quantification over latent embeddings and supports integration into hierarchical models. These properties have made BMDS a valuable tool in applications such as viral phylogeography \citep{bedford2015, holbrook2021bigbmds, li2023}, incorporation of auxiliary information \citep{lin2019} and clustering \citep{BMDScluster, sheth2026}.

Despite these advantages, BMDS is computationally constrained by the $O(n^2)$ cost of likelihood evaluation, which requires summing over all pairwise distances \Eqref{eq:bmds_ll}. As $n$ grows, this quadratic scaling becomes prohibitive, restricting BMDS inference to relatively small datasets. Existing approximation methods \citep{sheth2026} lack rigorous posterior consistency guarantees for two or more dimensional latent spaces, raising concerns about their practical reliability. To address these limitations, we propose an approximation to BMDS based on the Barnes--Hut algorithm (BH-BMDS). This framework scales as $O(n \log n)$, preserves the original Bayesian formulation and guarantees the consistency for the stationary measure under the approximation and its stochastic realizations.

\section{Related Work}
Several approaches attempt to mitigate the cost of BMDS. \cite{bedford2015} adopt a misspecified Gaussian likelihood to bypass the computational bottleneck imposed by non-negativity constraints; however, this assumption places probability mass outside the support of the observed data. \cite{holbrook2021bigbmds} leverage multi-core central processing units, vectorization and graphic processing units to significantly reduce runtime but depend on specialized hardware for scalability. Both methods require $O(n^2)$ floating point operations. To achieve sub-quadratic scaling, \cite{sheth2026} induce sparsity in the dissimilarity matrix, achieving $O(mn)$ complexity, where $m \ll n$ represents the number of landmarks. Although this approach yields substantial speedups, $m$ must scale with $n$ to maintain accuracy, and its selection introduces an additional tuning parameter. Moreover, they do not provide posterior consistency guarantees for subsampled dissimilarity matrices in dimensions $d > 1$. 

Tree-based methods offer an alternative strategy for approximating pairwise interactions. For instance, \cite{tsne2014} use the Barnes--Hut algorithm to accelerate t-SNE by approximating long-range interactions in an $n$-body system. While this scales as $O(n\log n)$, the paper lacks formal accuracy bounds for the approximation error. Noting that the BMDS likelihood admits a similar $n$-body interpretation, we adapt the Barnes--Hut algorithm to the BMDS setting (BH-BMDS), enabling $O(n \log n)$ likelihood approximations within an MCMC framework. We further establish consistency under uniform vanishing pairwise loss error for BH-BMDS, proving the stationary measure concentrates on the true configuration despite errors introduced by the approximation and noisy MCMC updates.

\section{Preliminaries} \label{sec:methods}
\subsection{Bayesian multidimensional scaling} \label{sec:BMDS}
Suppose we have high-dimensional observations, e.g., $\highdata_1, \dots \highdata_n \in \RR^{\Omega}$, and let $D_{ij} = \mathcal{D}(\mathbf{y}_i, \mathbf{y}_j)$ denote the dissimilarity between observations. Bayesian multidimensional scaling (BMDS) \citep{ohraftery2001bmds} assumes that the observed dissimilarities follow
\begin{align}\label{eq:bmds_model}
    D_{ij} \;=\; \big[\,r_{ij} + \sigma Z_{ij}\,\big]_+,
    \qquad r_{ij} = \norm{\latentdata_i -\latentdata_j},
    \qquad Z_{ij}\stackrel{\text{i.i.d.}}{\sim}N(0,1),
\end{align}
independent across pairs and conditioned on the low-dimensional latent locations $\latentData = \{\latentdata_i\}^n_{i = 1} \subset \RR^d$. Here, $[t]_+ = \max\{t, 0\}$ represents truncation at zero, and the BMDS variance $\sigma^2$ represents the level of noise distorting the latent distances. This yields the log-likelihood
\begin{align} \label{eq:bmds_ll}
    \ell(\distanceMatrix \mid \latentData, \sigma^2) = -\frac{N}{2} \log(2\pi\sigma^2) - \sum_{i < j} \left[ \frac{(D_{ij} - r_{ij})^2}{2\sigma^2} + \log \Phi \left( \frac{r_{ij}}{\sigma} \right) \right], 
\end{align} 
where $\distanceMatrix = \{D_{ij}\}$ is the observed symmetric $n \times n$ dissimilarity matrix, $N = \binom{n}{2}$, and $\Phi(\cdot)$ is the standard normal CDF. Because the likelihood involves a summation over all unordered pairs, its evaluation scales as $O(n^2)$. Consequently, we seek to reduce this computational burden using tree-based approximation methods, such as Barnes–Hut algorithms, that scale as $O(n\log n)$.

\subsection{Barnes--Hut algorithm} \label{sec:BH-algorithms}
\begin{figure}
\centering
\begin{subfigure}[c]{0.3\textwidth}
    \centering
    \includegraphics[width=0.9\linewidth]{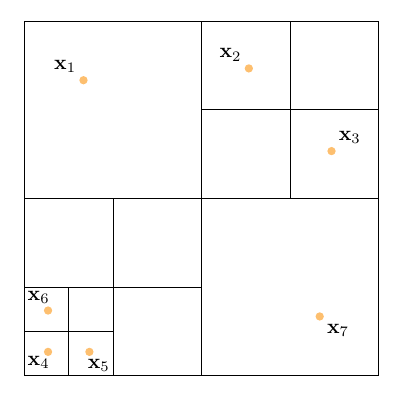}
    %\caption{Spatial representation}
\end{subfigure}
\begin{subfigure}[c]{0.3\textwidth}
    \centering
    \includegraphics[width=0.9\linewidth]{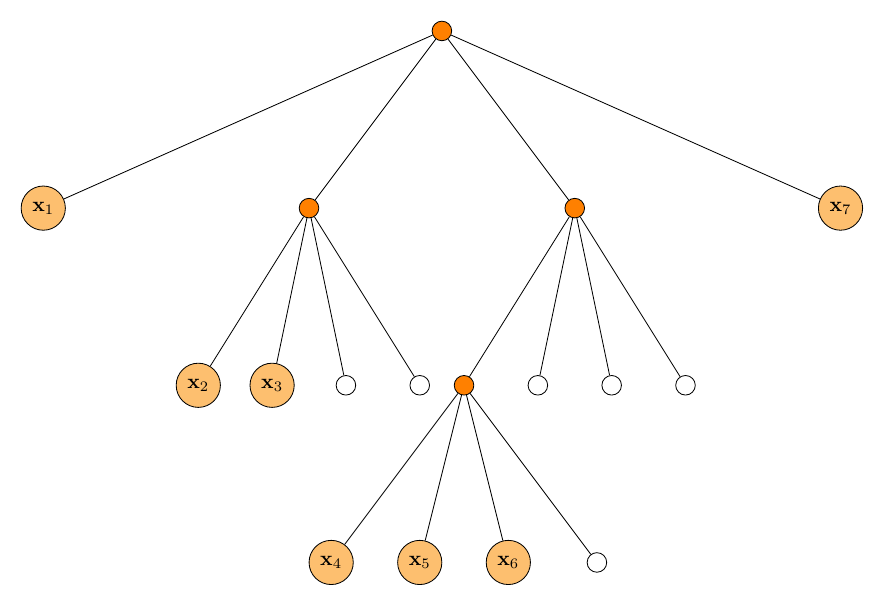}
    %\caption{Quadtree}
\end{subfigure}
\caption{Barnes--Hut structural components. The left figure is the spatial representation of a recursive subdivision of the embedding space into four quadrants while the right figure is the respective quadtree data structure. One builds the quadtree using Algorithm~\ref{alg:QT}.}
\label{fig:BH-QT}
\end{figure}

The gravitational $n$-body problem involves computing the gravitational forces exerted by all pairs of particles on one another. Direct evaluation of these forces requires $O(n^2)$ computations which becomes prohibitive for large datasets. To reduce this cost, \cite{barneshut1986} propose a hierarchical approximation algorithm that estimates the influence of distant particles as a single aggregate interaction.

To organize groups of particles efficiently, the Barnes--Hut algorithm uses a tree data structure, e.g., a quadtree $\mathcal{T} \in \RR^2$ (Figure \ref{fig:BH-QT}). A quadtree $\mathcal{T}$ is built by recursively subdividing the embedding space into four quadrants based on particle locations (Algorithm \ref{alg:QT}). A cell refers to a fixed spatial region, and a node corresponds to a cell in the tree. Each leaf node contains at most one particle, while internal nodes store summary statistics of all particles within the cell as a pseudo-particle. As particles are inserted, these summaries are updated recursively. The quadtree $\mathcal{T}$ is constructed in $O(n \log n)$ time and enables efficient spatial queries and interaction approximations.

To compute interactions, the quadtree $\mathcal{T}$ is traversed starting from the root node. At each node $k \in \{0, \dots, K\}$ where the total number of nodes $K = O(n)$, we consider its corresponding cell and pseudo-particle. Let $w_k$ be the width of this cell, $\bar{r}_{ik}$ the Euclidean distance between the target particle $\latentdata_i$ and the pseudo-particle's center of mass, and $\rho_{ik} = w_k/\bar{r}_{ik}$. If $\rho_{ik} < \theta$, where $\theta$ is a user-specified threshold, the pseudo-particle is considered distant and used as an approximation of its constituent particles. Otherwise, we recursively examine the node's children; interactions with particles in leaf nodes are computed directly. The parameter $\theta$ controls the trade-off between accuracy and computational efficiency; larger values of $\theta$ result in fewer tree traversals and faster computation at the cost of reduced accuracy. $\theta \ge 0.5$ typically yields $O(n \log n)$, while $\theta = 0$ recovers the exact $O(n^2)$ computation.

\section{A tree-based framework for Bayesian multidimensional scaling} \label{sec:BH-BMDS}
\subsection{Barnes--Hut Bayesian multidimensional scaling}
Given latent embeddings $\{\latentdata_i\}^n_{i = 1} \subset \RR^2$, we construct a quadtree $\mathcal{T}$ (Algorithm \ref{alg:QT}) where each node $k$ corresponds to a spatial cell $C_k$ in the latent space with index set $\mathcal{I}_k = \{i \in \{1, \dots, n\} : \latentdata_i \in \mathcal{C}_k \}$. While leaf nodes contain single observations, internal nodes store pseudo-particles summarizing the number of points ($N_k = |\mathcal{I}_k|$) and the centers of mass of the high-dimensional ($\bar{\highdata}_k = \frac{1}{N_k} \sum_{l \in \mathcal{I}_k} \highdata_l$) and low-dimensional ($\bar{\latentdata}_k = \frac{1}{N_k} \sum_{l \in \mathcal{I}_k} \latentdata_l$) points within $\mathcal{C}_k$. The root node ($k = 0$) summarizes the entire dataset. 

Unlike t-SNE \citep{tsne2014}, BMDS \Eqref{eq:bmds_ll} depends on squared differences between high-dimensional and low-dimensional dissimilarities. To preserve the $O(n \log n)$ complexity, we approximate the high-dimensional dissimilarity $
\bar{D}_{ik} = \lVert \highdata_i - \bar{\highdata}_k \rVert
$ using the Euclidean distance between the observed point $\highdata_i$ and the pseudo-particle’s high-dimensional center of mass $\bar{\highdata}_k$. This choice avoids storing or computing all pairwise dissimilarities which would require $O(n^2)$ memory or operations, at the cost of introducing approximation error.

During likelihood computations, we perform a depth-first traversal of quadtree $\mathcal{T}$ for each point $\latentdata_i$ and compute
$
\bar{r}_{ik} = \lVert \latentdata_i - \bar{\latentdata}_k \rVert, 
$
the Euclidean distance between the target point $\latentdata_i$ and the pseudo-particle’s low-dimensional center of mass $\bar{\latentdata}_k$. This quantity is used within the Barnes--Hut criterion to determine tree traversal. The Barnes--Hut BMDS (BH-BMDS) framework reduces the computational burden of evaluating the BMDS log-likelihood through the use of summary statistics; dissimilarities between a target point and points within a sufficiently distant cell can be represented by their aggregate summary. 

Because each unordered pairwise interaction is visited twice (one per target point), the normalization term is adjusted by a factor of two. The BH-BMDS log-likelihood is
\begin{align} \label{eq:bh_bmds_ll}
    \tilde{\ell}(\distanceMatrix \mid \latentData, \sigma^2, \mathcal{T})
    &= -N \log(2\pi\sigma^2)
    - \sum_{i=1}^n \sum_{k \in \mathcal{B}_i}
    N_k \biggl[
        \frac{(\bar{D}_{ik} - \bar{r}_{ik})^2}{2\sigma^2}
        + \log \Phi \!\left( \frac{\bar{r}_{ik}}{\sigma} \right)
    \biggr], \nonumber \\ 
    &\equiv -N \log(2\pi\sigma^2)
    - \sum_{i=1}^n \tilde{\ell}(\latentdata_i \mid \mathcal{T}),
\end{align} 
where $\mathcal{B}_i$ denotes the set of nodes satisfying the Barnes--Hut acceptance criterion (or leaf nodes) for target point $\latentdata_i$, and $\tilde{\ell}(\latentdata_i \mid \mathcal{T})$ is the pointwise contribution to the BH-BMDS log-likelihood from a Barnes--Hut traversal of $\latentdata_i$ starting at the root node (Algorithm~\ref{alg:deter}). The BH-BMDS log-likelihood approximates the full BMDS log-likelihood, and when $\theta = 0$, it reduces to the exact formulation since $\mathcal{B}_i$ contains only leaf nodes.

\subsection{Noisy Barnes--Hut algorithm} \label{sec:stochBH}
Under the Barnes--Hut approximation, the BMDS log-likelihood becomes a piecewise function of the latent positions due to the discrete tree structure. For a fixed set of points, small changes in the embeddings can alter traversal paths or cell assignments, leading to discontinuous changes in the BH-BMDS log-likelihood. Therefore, we introduce a probabilistic acceptance rule to smooth the transition of accepting pseudo-particles near the threshold $\theta$, mitigating abrupt changes in the Barnes--Hut approximation as $\rho_{ik}$ varies for small movements in $\latentdata_i$.

Noisy Barnes--Hut follows the same tree traversal as the deterministic version but relaxes the hard acceptance criterion. Instead of the binary indicator $\mathbb{I}(\rho_{ik} < \theta)$, noisy Barnes--Hut treats this decision probabilistically. A pseudo-particle is accepted if $u < p(\rho_{ik} \mid \theta)$, where $u \sim \text{Unif}(0,1)$ and the acceptance probability is defined by the logistic function
\begin{align}\label{eq::logfunc}
    p(\rho \mid \theta) = \left( 1 + \exp\left[ -\frac{c}{\theta}(\theta - \rho) \right] \right)^{-1}.
\end{align}
The parameter $\theta$ is a user-specified threshold, while $c$ is a constant controlling the slope of the logistic function. We fix $c$ so that the transition between acceptance and rejection occurs over a narrow band of $\rho_{ik}$ values; specifically, $p(\rho \mid \theta) = 0.8$ when $(\theta - \rho)/\theta = 0.2$. In our implementation, this corresponds to $c = 6.93$. As $c \to \infty$, the probability $p(\rho \mid \theta)$ converges to the indicator function $\mathbb{I}(\rho < \theta)$, recovering the deterministic rule. The method preserves the expected computational complexity of $O(n \log n)$ with occasional additional node evaluations near the acceptance boundary.

By treating $\mathcal{B}_i$ as a realization of a random process with controlled variance, the stochastic approach improves mixing in MCMC settings (Section \ref{sec:BayesianComp}). However, it introduces approximation noise that must be rigorously accounted for. We address this by proving that the stationary measure of BH-BMDS is consistent (Section \ref{sec:PostCons}).

\subsection{Consistency for noisy BH-BMDS} \label{sec:PostCons}
%\textcolor{red}{[AMS: We do prove that the posterior is consistent, but more importantly we prove that the stationary measure of the algorithm is consistent despite being (pointwise) very far from the posterior. I'd change the title of this section, and probably propagate the change to places where we refer to consistency.]}
We prove consistency for a class of surrogate-likelihood algorithms, even in the regime where the surrogate is not pointwise accurate. See Section \ref{app:main_prop} for an expanded discussion of this theoretical framework. 

Fix dimension $d \in \mathbb{N}$ and noise level $\sigma > 0$. For each $n \in \mathbb{N}$, we assume the data are generated from "true" latent positions $\mathbf{x}_1^\star, \dots, \mathbf{x}_n^\star \stackrel{\text{i.i.d.}}{\sim} \mathrm{Unif}([0,1]^d)$, and the observed dissimilarities $D_{ij}$ are generated according to Model \ref{eq:bmds_model} 
%\textcolor{red}{[AMS: I'm OK with this, but personally I would prefer to make a much smaller distinction between capping observations at 0 and capping them at $\sqrt{M}$ - in both cases, we're just saying that the observations can't be obviously inconsistent with the underlying data-generating process for the latent positions. Making this change also clears up a lot of the later discussion - we're not talking about a funny proxy posterior based on truncation, we're talking about the actual posterior you get once you fix up obviously-wrong observations in a very conservative way.]}, 
where $r_{ij} = r_{ij}^\star \equiv \norm{\latentdata_i^\star - \latentdata_j^\star}$. 
%We establish posterior consistency for noisy MCMC by studying a surrogate risk function, a proxy for the posterior distribution, using capped observations $\tilde{D}_{ij}$.  \textcolor{red}{[AMS: I think this confuses objects a bit - we refer to some things that aren't posteriors as posteriors, and other things that ARE posteriors as surrogates. A possible replacement: We show that the stationary measure for noisy MCMC gives a consistent estimator for the ground truth, even though the surrogate likelikelihoods are pointwise quite far from the true posterior.]}
We define the truncation function $T_M(u)=\min\{u,M\}$ for $u\ge 0$, and a maximum distance $M \ge \sqrt d$ . We then set the capped dissimilarities and their expected values as
\begin{align}\label{eq:noisyD}
    \tilde D_{ij}=T_M(D_{ij}),
    \qquad
    \mu_M(r)=\E\big[T_M([r+\sigma Z]_+)\big],\qquad r\ge 0
\end{align}
such that the conditional law of $\tilde D_{ij}$ given $r_{ij}^\star=r$ has a piecewise density $q_{M,r}(y)$ \Eqref{eq:cappedDensity}. Since the maximum possible distance within this unit hypercube is bounded by $\sqrt d$, the capped distances are consistent with the underlying data-generation process. 

For a candidate configuration $\latentData =(\latentdata_1,\dots,\latentdata_n)\in([0,1]^d)^n$ and $r_{ij}(\latentData)=\norm{\latentdata_i-\latentdata_j}$, define the loss (negative log-likelihood) as %the exact capped loss as \textcolor{red}{[AMS: I'm sure I introduced it, but the ``exact capped loss" is just the log-likelihood. Since we're hurting for space, it probably makes sense to squeeze a bunch of the next few lines - this is just the ``usual" Bayesian machine.]}
\begin{align}\label{eq:loss}
\ell_M(y,r) 
\;=\;
-\log q_{M,r}(y),
\qquad y\in[0,M],\ r\in[0,\sqrt d],
\end{align} %Letting $N = \binom{n}{2}$ denote the total number of pairs \textcolor{red}{[AMS: We probably don't need notation for $N$ here, as it is only used once. Please put it in the appropriate part of the appendix if we remove it here.]}, we define 
and the associated empirical risk as
\begin{align}\label{eq:empRisk}
R_n(\latentData) \;=\; \frac1N\sum_{1\le i<j\le n} \ell_M(\tilde D_{ij},\,r_{ij}(\latentData)).
\end{align}

%To capture the stochasticity of our noisy MCMC, let $\Xi$ denote the space of algorithmic realizations \textcolor{red}{[AMS: This might be my fault too, but ``space of algorithmic realizations" seems both verbose and imprecise to me. I think we actually want to flip the order of the objects in this paragraph, as some of the rest of the paragraph doesn't really make sense as-written. What we actually have are a collection of surrogate log-likelihoods indexed by $\xi$. These induce algorithms, and a measure on $\Xi$ induces a noisy MCMC algorithm. Here we introduce $\Xi$ first and try to claim that this comes with a unique associated surrogate log-likelihood, but that's definitely not true. ]}. Let $\xi_0 \in \Xi$ represent the deterministic reference algorithm, e.g., standard Barnes--Hut, and let $\xi \in \Xi$ be a stochastic realization. 
We consider a family of surrogate log-likelihoods indexed by $\xi \in \Xi$, where each $\xi$ induces an MCMC algorithm. Let $\xi_0 \in \Xi$ represent the deterministic reference algorithm, e.g., standard Barnes--Hut. By defining a probability measure on $\Xi$, we draw a stochastic realization $\xi$, yielding the randomized surrogate log-likelihood that drives our noisy MCMC. We denote $\widehat\ell_{ij}^{(\xi)}(\latentData)$ as the resulting approximation of $\ell_M(\tilde D_{ij},r_{ij}(\latentData))$ and define the surrogate risk as
\begin{align}\label{eq:approxRisk}
\widehat R_n^{(\xi)}(\latentData) \;=\; \frac1N\sum_{1\le i<j\le n} \widehat\ell_{ij}^{(\xi)}(\latentData).
\end{align}
Note that $\widehat\ell_{ij}^{(\xi)}$ depends on the entirety of $\latentData$ due to the quadtree construction. 

Since distances are invariant under Euclidean isometries, e.g., rigid motion, we measure convergence using the orbit distance $d_{\cG}$ defined in Remark \ref{rmk1:OD}. Specifically, we prove that the mass of the stationary measure outside any fixed $d_{\cG}$-ball around the true configuration $\latentData^\star$ goes to $0$ (in probability) as $n\to\infty$. To achieve this, we assume a standard data-generating process characterized by a random design, independent paired noise and sufficient prior regularity around the true configuration (Assumption~\ref{app:A:model}). The key assumption is \emph{uniform vanishing error} for the \emph{pairwise losses}:
\begin{assumption} [Small Approximation Error] \label{A:error}
There exists a deterministic sequence $\varepsilon_n\downarrow 0$ such that, with probability $\to 1$ under the joint law of $(\latentData^\star, \mb{D})$, the reference approximation $\widehat \ell^{(\xi_0)}$ satisfies 
\begin{align}\label{eq:pairwiseError}
\sup_{\latentData\in([0,1]^d)^n}\ \max_{1\le i<j\le n}\ \Big|\widehat\ell_{ij}^{(\xi_0)}(\latentData)-\ell_M(\tilde D_{ij},r_{ij}(\latentData))\Big|
\ \le\ \varepsilon_n.
\end{align}
\end{assumption}

While the difference in the pairwise approximation error must go to 0 as $n$ grows, this assumption allows the difference in the total log-likelihood to diverge to infinity, provided the growth remains sub-quadratic. 
%Note that, on the high-probability event \ref{eq:pairwiseError}, the error in the risk is also small \textcolor{red}{[AMS: If we want to include this in the main text, I'd give exactly the opposite emphasis. Yes, the difference in risk goes to 0... but the difference in log-likelihoods is allowed to go to infinity as long as it is sub-quadratic. This sort of analysis isn't unique, but as far as I know this is the first such analysis for MCMC algorithms with nonparametric/infinite-dimensional samples to learn.]}
%\begin{align}\label{eq:riskError}
%\sup_{\latentData\in([0,1]^d)^n}\big|\widehat R_n^{(\xi_0)}%(\latentData)-R_n(\latentData)\big|\le \varepsilon_n.
%\end{align}

\begin{theorem}[Consistency for noisy MCMC]\label{thm:NoisyPC}
Fix $d\ge 1$, $\sigma>0$, $M\geq \sqrt d$ and a parameter space $\Xi$. Suppose there exists a reference parameter $\xi_0 \in \Xi$ such that the chain $P_{\xi_0}$ is uniformly ergodic and let $\widehat{\ell}_{ij}^{(\xi_0)}$ be an approximation satisfying Assumption \ref{A:error}. Let $P_{\xi}$ be a Markov chain with stationary measure $\widehat{\Pi}_{n}^{(\xi)}$, and define the mixture kernel $P=\int_{\Xi} P_{\xi}\,\mu(d\xi)$ for a probability measure $\mu$ on $\Xi$. Let $\widetilde\Pi_n$ be an invariant probability measure of $P$ and set $\kappa_n$ and $\Delta_n$ as Equations \ref{app:thm_lkdef} and \ref{app:thm_deltadef}, respectively. If $\kappa_n\Delta_n\xrightarrow{\Pbb}0$, then for every $\delta > 0$,
\begin{align}
\widetilde\Pi_n\Big(\big\{\latentData:\ d_{\cG}(\latentData, \latentData^\star)>\delta\big\}\Big)\xrightarrow{\Pbb}0.
\end{align}
\end{theorem}

%\textcolor{red}{[AMS: We probably want to mention the orbit distance early on, with a forward-reference to the definition; at the moment we never say what the symbol is, though I can see the first sentence of the remark is supposed to be about this. The remark itself can probably be shortened a little. ]}

\begin{remark}
    Theorem \ref{thm:NoisyPC} applies directly to the deterministic Barnes--Hut algorithm ($\xi = \xi_0$, Section~\ref{sec:BH-algorithms}, Algorithm~\ref{alg:deter}) by choosing a sequence $\theta = \theta_n \to 0$. More subtly, the result also extend to the noisy quadtree algorithm (Section~\ref{sec:stochBH}, Algorithm~\ref{alg:stoch}) for sufficiently small noise. Because the target distribution is supported on a compact state space and possesses a density strictly bounded away from both zero and infinity, utilizing any standard proposal distribution with a similarly bounded density ensures that the reference Metropolis-Hastings chain is uniformly ergodic. Furthermore, the one-step total variation distance $\Delta_n$ is governed by the logistic slope $c$ in the noisy traversal \Eqref{eq::logfunc}. By letting $c \to \infty$, $\Delta_n$ can be made arbitrarily small, fulfilling the final theoretical assumptions.
\end{remark}

\paragraph{Proof Sketch} The full proof, detailed in Appendix \ref{app:PCproof}, proceeds in two main stages. First, we establish consistency for the deterministic reference ($\xi_0$). We demonstrate that the baseline empirical risk uniformly concentrates around the expected "population" risk. Because this population risk is strictly minimized at the true configuration $\latentData^\star$ (up to isometry), the mass of the stationary measure must concentrate around $\latentData^\star$ in orbit distance. Second, we apply perturbation arguments to show this consistency is preserved under the noisy quadtree algorithm. 

\subsection{Bayesian computation} \label{sec:BayesianComp}
% alg 4 - MCMC sampler
\begin{wrapfigure}[21]{r}{0.45\textwidth}
\vspace{-20pt}
    \begin{minipage}[t]{\linewidth}
        \begin{algorithm}[H]
        \footnotesize
        \caption{BH-BMDS via \\ Metropolis-within-Gibbs}\label{alg:MCMC}
        \KwIn{$\highdata_1,\dots, \highdata_n \in \RR^\Omega$}
        \KwOut{posterior samples for $\latentdata_1,\dots,\latentdata_n \in \RR^2$}
        
        Initialize $\latentData^0 = \{\latentdata_i^0 \}^{n}_{i = 1} \subset \RR^2$ via dimension reduction (e.g., MDS or PCA)\;
        Build quadtree $\mathcal{T}$ on $\latentData^0$\;
        
        \For{each MCMC iteration}{
            Permute indices $1,\dots,n$\;
            \For{$i = 1,\dots,n$}{
                Compute current $\tilde{\ell}(\latentdata_i \mid \mathcal{T})$ using Algorithm~\ref{alg:stoch}\;
                Propose $\latentdata_i' \sim N(\latentdata_i, \tau_i^2 \mb I)$\;
                Update $\mathcal{T}$ locally using Algorithm~\ref{alg:QT_updates}\;
                Compute proposed $\tilde{\ell}(\latentdata_i' \mid \mathcal{T})$ using Algorithm~\ref{alg:stoch}\;
                Accept/reject via MH criterion \Eqref{eq:MHAR}\;
                \If{rejected}{
                Restore cached node summaries in $\mathcal{T}$\;
                }
                \Else{
                Keep updated $\mathcal{T}$\;
                }
            }
            Compute $\tilde{\ell}(\distanceMatrix \mid \latentData, \sigma^2, \mathcal{T})$ using Algorithm~\ref{alg:deter}\;
            Propose $1/\sigma^{2\,\prime} \sim N_{(0, \infty)}(1/\sigma^{2(s)}, \tau_{1/\sigma^{2}}^{2(s)})$\;
            Compute $\tilde{\ell}(\distanceMatrix \mid \latentData, \sigma^{2\,\prime}, \mathcal{T})$ using Algorithm~\ref{alg:deter}\;
            Accept/reject via MH criterion \Eqref{eq:MHAR}\;
            Rebuild $\mathcal{T}$\
        }
        \end{algorithm}
    \end{minipage}
\end{wrapfigure}
To infer the model parameters, we implement a Metropolis-within-Gibbs (MwG) sampler (Algorithm~\ref{alg:MCMC}) targeting the posterior distribution
\begin{align}\label{eq:postdist}
    p(\latentData, \sigma^2 \mid \mathcal{T})
    \propto
    \tilde{\ell}(\distanceMatrix \mid \latentData,  \sigma^2, \mathcal{T}) \pi(\latentData)\,\pi(\sigma^2),
\end{align}
where $\pi(\cdot)$ is the prior specification. 

The sampler proceeds by sequentially updating each model parameter $\gamma = (\latentdata_1,\dots, \latentdata_n, \sigma^{2})$ using Metropolis--Hastings (MH) steps within a Gibbs framework. At iteration $s$ and for a generic parameter $\gamma_k$ with proposal $\gamma_k' \sim q(\cdot \mid \gamma_k^{(s)})$, the acceptance probability is
\begin{align}\label{eq:MHAR}
    \alpha(\gamma_k' \mid \gamma^{(s)}, \mathcal{T}) = \nonumber \\ \min \left\{1,
    \frac{
        p(\gamma_k', \gamma_{-k}^{(s)} \mid \mathcal{T})\;
        q(\gamma_k^{(s)} \mid \gamma_k')
    }{
        p(\gamma_k^{(s)}, \gamma_{-k}^{(s)} \mid \mathcal{T})\;
        q(\gamma_k' \mid \gamma_k^{(s)})
    } \right\},
\end{align}
where $\gamma_{-k}$ denotes all parameters except $\gamma_k$.
We propose a candidate latent location from 
$\latentdata_n' \sim N\!\big(\latentdata_n^{(s)}, \tau_{n}^{2(s)} \mb I\big)$
and a candidate precision from a truncated normal proposal,
$1/\sigma^{2\,\prime} \sim N_{(0, \infty)}(1/\sigma^{2(s)}, \tau_{1/\sigma^{2}}^{2(s)})$.

To efficiently evaluate the posterior, we implement different Barnes--Hut approximations depending on the parameter being updated. Latent locations use Algorithm~\ref{alg:stoch} which evaluates a pointwise contribution to the BH-BMDS likelihood via a stochastic tree traversal. Noisy Barnes--Hut mitigates discontinuities in the likelihood approximation that arise when small changes in $\latentdata_i$ cause crossings of the deterministic Barnes--Hut acceptance boundary $\mathbb{I}(\rho_{ik} < \theta)$. By replacing this hard threshold with a probabilistic acceptance rule \Eqref{eq::logfunc}, the resulting likelihood approximation varies more smoothly as a function of $\latentdata_i$ which increases the stability of Metropolis-Hastings updates. In contrast, the variance parameter $\sigma^{2}$ uses deterministic Barnes--Hut to approximate the full BMDS likelihood (Algorithm~\ref{alg:deter}). This strategy avoids introducing additional noise into likelihood evaluations, thereby improving numerical stability.

Under this MwG sampler, for fixed quadtree $\mathcal{T}$, a single latent location update requires $O(\log n)$ expected time, and a full sweep over all locations costs $O(n \log n)$. We assume proposals are sufficiently local such that $\latentdata_i$ and $\latentdata_i'$ lie in the same quadtree path, so that only center-of-mass updates are required. While running experiments on data (Section \ref{sec:results}), we found that this assumption holds true, and within a sweep, points' cell allocations do not change. At the end of each sweep, the quadtree $\mathcal{T}$ is rebuilt using the updated latent locations, which also requires $O(n \log n)$ time. The cost of updating $\sigma^{2}$ is of the same order as a full sweep.

\section{Experiments} \label{sec:results}
We conduct experiments on different datasets to evaluate the performance of Barnes--Hut BMDS. All algorithmic results arise from single-core CPU implementations carried out on a single Linux machine with two 26-core Intel Xeon Gold processors (2.1 GHz). Tree related data structures and likelihood computations are implemented in \texttt{C} while MCMC samplers are implemented in \texttt{Python}. For visualization, we use the \texttt{ggplot2} \citep{ggplot2} package in \texttt{R} \citep{Rcite}.

\subsection{Data}
We describe each of the four datasets that span structured mobility data, text embeddings, image representations and neural activity recordings, allowing us to evaluate BH-BMDS across diverse modalities. We also explain our preprocessing step for defining high-dimensional inputs which are mapped to two-dimensional latent representations for BH-BMDS. Dimensions are chosen to balance information retention and computational efficiency. 

\paragraph{Influenza} This dataset consists of air traffic-dervied pairwise "effective distances" between $n = 5{,}392$ taxa. The taxa belong to four influenza subtypes: $1{,}370$, $1{,}389$, $1{,}393$ and $1{,}240$ observations of H1N1, H3N1, Victoria and Yamagata, respectively. Effective distances inversely measure the probability of traveling between airports, such that shorter distances correspond to more heavily trafficked routes. Compared to geometric distances, effective distances better capture the structure of the global mobility network and are therefore more appropriate for modeling disease spread \citep{brockhelbing2013}. Since multiple taxa are associated with the same country, the dataset exhibits strong clustering, with only $66$ unique inter-country distances. As preprocessing, we apply classical MDS to the dissimilarity matrix to obtain Euclidean embeddings $\highdata_i \in \RR^6$.

\paragraph{ArXiv} This dataset consists of $n = 10{,}000$ abstracts collected from arXiv, each labeled with one of five subject areas: math (\texttt{math}), astrophysics (\texttt{astro-ph}), optics (\texttt{physics.optics}), quantum physics (\texttt{quant-ph}) and computer science / machine learning (\texttt{cs/ml}). Each abstract is embedded into 768-dimensional vector $\mathbf{w}_i$ using \texttt{SentenceTransformers} \citep{sentence-trans} under the \texttt{all-mpnet-base-v2} model \citep{song2020mpnet} in \texttt{Python}. This large language model captures semantic similarity at the sentence level. We compute the pairwise cosine dissimilarities, e.g., $1 - \cos(\mathbf{w}_i, \mathbf{w}_{j}) = 1 - \frac{\mathbf{w}_i \cdot \mathbf{w}_{j}}{\|\mathbf{w}_i\| \|\mathbf{w}_{j}\|}$ for $\mathbf{w}_i \in \RR^{768}$ between embeddings and then apply classical MDS to obtain Euclidean embeddings $\highdata_i \in \RR^{50}$.

\paragraph{MNIST} This dataset is a random subset $n = 10{,}000$ of grayscale handwritten digit images from the MNIST dataset, where each image has dimension $d = 28 \times 28 = 784$ with pixel intensity in $[0,1]$ and labels from 0 to 9. To extract meaningful features, we employ a convolutional autoencoder that maps each image to a 30-dimensional latent representation $\highdata_i \in \RR^{30}$. The model is trained to minimize pixel-wise reconstruction loss using the Adam optimizer in \texttt{PyTorch} \citep{pytorch}. The resulting 30-dimensional embeddings serve as the high-dimensional input for BH-BMDS.

\paragraph{Neural activity} Data come from experiments reported in \cite{alz_expt}. This dataset comprises local field potential (LFP) recordings from 32-channel probes implanted in the CA1 hippocampus and medial entorhinal cortex of healthy and ThyTau22 mice. Recordings were captured over six months during a goal-oriented virtual reality task that lasted 15 minutes. We employ a pre-trained MaGNet GNN \citep{magnet}--trained to decode behavioral states (stalled vs.~moving)--to generate high-dimensional latent embeddings $\highdata_i \in \RR^{51}$, independently across mouse types and weeks. We use BH-BMDS to project these embeddings into a low-dimensional space to evaluate manifold interpretability and identify progressive signs of dementia.

\subsection{Experimental setup}\label{sec:setup}
For our simulations, we assume that our posterior distribution \Eqref{eq:postdist} is specified with priors $\latentdata_1,\dots,\latentdata_n \stackrel{\text{i.i.d.}}{\sim} N(\mb 0, \mb I_2)$ and $1/\sigma^{2} \sim \text{Gamma}(1,1),$
where $1/\sigma^{2}$ is the BMDS precision parameter. The proposal variances $\tau_{i}^{2(s)}$ and $\tau_{\sigma^{-2}}^{2(s)}$ are adaptively tuned to maintain target acceptance rates, thereby facilitating efficient exploration of the posterior. Let $\hat{\alpha}$ denote the empirical acceptance rate over a fixed adaptation window and $\alpha$ the target acceptance rate. Define $R = \hat{\alpha}/\alpha$, truncated to lie in $[0.5, 2]$. After each adaptation window, the proposal variances are updated via $\tau^{(s+1)} = R \tau^{(s)}$. Adaptation is performed every 50 MCMC iterations. We set $\theta = 2$ for all experiments as it empirically provides a favorable trade-off between accuracy and computational efficiency.

\subsection{Accuracy} \label{sec:acc}
Building on the theoretical guarantees for consistent posterior behavior (Theorem~\ref{thm:NoisyPC}), we empirically evaluate BH-BMDS across four datasets. We visualize the posterior Procrustes-aligned mean embeddings using every 50th MCMC iteration. Across all benchmarks (Figures~\ref{fig:flu_all}--\ref{fig:alz_all}), BH-BMDS closely recovers the clustering structure of the exact BMDS, with shared labels consistently co-located. We observe minor coordinate shrinkage toward the origin, a known effect of Gaussian priors when the likelihood approximation introduces additional variance. Notably, exact BMDS results are only provided for smaller subsets where computation is feasible; for larger $n$, exact inference is prohibitive. These results demonstrate that BH-BMDS provides a faithful, scalable approximation to the full Bayesian formulation.

\begin{figure}
    \centering
    \includegraphics[width=\linewidth]{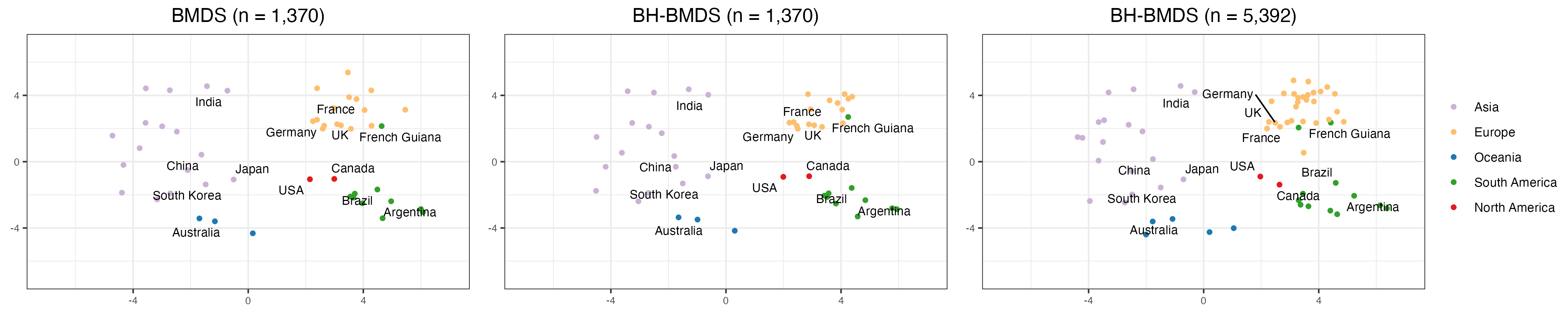}
    \caption{Posterior Procrustes-aligned means across MCMC iterations and taxa for only H1N1 ($n = 1{,}370$) and all influenza subtypes ($n = 5{,}392$) under full Bayesian multidimensional scaling (BMDS) and Barnes--Hut Bayesian multidimensional scaling (BH-BMDS) frameworks.}
    \label{fig:flu_all}
\end{figure}

\begin{figure}
    \centering
    \includegraphics[width=\linewidth]{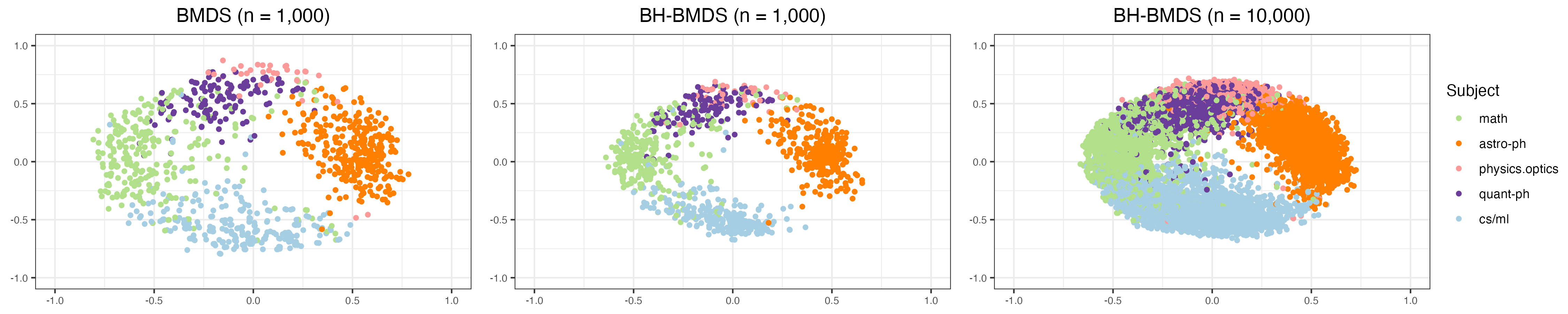}
    \caption{Posterior Procrustes-aligned means across MCMC iterations for $1{,}000$ and $10{,}000$ embedded ArXiv abstracts under full Bayesian multidimensional scaling (BMDS) and Barnes--Hut Bayesian multidimensional scaling (BH-BMDS) frameworks. Figures \ref{fig:mnist_all} and \ref{fig:alz_all} show analogous results for the MNIST and neural activity datasets.}
    \label{fig:arxiv_all}
\end{figure}

\subsection{Computational Efficiency}\label{sec:compeff}
We evaluate efficiency using effective sample size (ESS) per hour, computed via the \texttt{coda} \citep{coda} package in \texttt{R}.
ESS inversely reflects the degree of autocorrelation in MCMC samples, with higher values indicating more efficient exploration of the posterior. We calculate ESS using thinned MCMC samples (retaining every 50th iteration) on a random subset of 1,000 pairwise latent distances to avoid identifiability issues inherent in raw locations.

Figure \ref{fig:eff_small} compares the computational efficiency across three inference frameworks: Hamiltonian Monte Carlo (HMC) and MwG for exact BMDS and Algorithm \ref{alg:MCMC} for BH-BMDS and five datasets: H1N1, ArXiv, MNIST and the ThyTau22 and Wild-Type mouse at Week 9. HMC efficiently explores the posterior by leveraging gradient information but incurs additional computational cost due to gradient evaluations \citep{neal2012}. In our implementation, HMC jointly samples latent locations under the full BMDS likelihood using 20 leapfrog steps. BMDS and BH-BMDS follow Algorithm \ref{alg:MCMC} with $\theta = 0$ and $\theta = 2$, respectively. For BH-BMDS, we repeat each experiment four times to account for variability induced by the quadtree-based likelihood approximation. As shown in Figure \ref{fig:eff_small}, BH-BMDS achieves approximately 60-fold speedup relative to HMC and 9-fold speedup relative to BMDS. Gains are most pronounced in datasets with distinct clustering; for instance, the ThyTau22 cohort shows slightly lower relative efficiency than Wild-Type due to more diffuse posterior mass (Figure \ref{fig:alz_all}).

As exact BMDS is infeasible at scale, we extrapolate its performance to the full datasets, generously assuming idealized sampling efficiency by setting $ESS = S$ for $S$ retained (thinned) samples. For the influenza dataset, we use observed BMDS runtimes; for the ArXiv and MNIST datasets, we estimate BMDS runtime at $n = 10{,}000$ via a quadratic regression on benchmarked runtimes ($n \in \{100, 500, 1{,}000$ and $5{,}000\}$). Even with this optimistic assumption, BH-BMDS achieves substantial gains: approximately 27-fold speedup for influenza, 67-fold for MNIST and 93-fold for ArXiv (Figure \ref{fig:eff_big}). These gains increase with $n$, consistent with the improved scaling behavior of the BH-BMDS approximation.

\begin{figure}
    \centering
    \includegraphics[width=\linewidth]{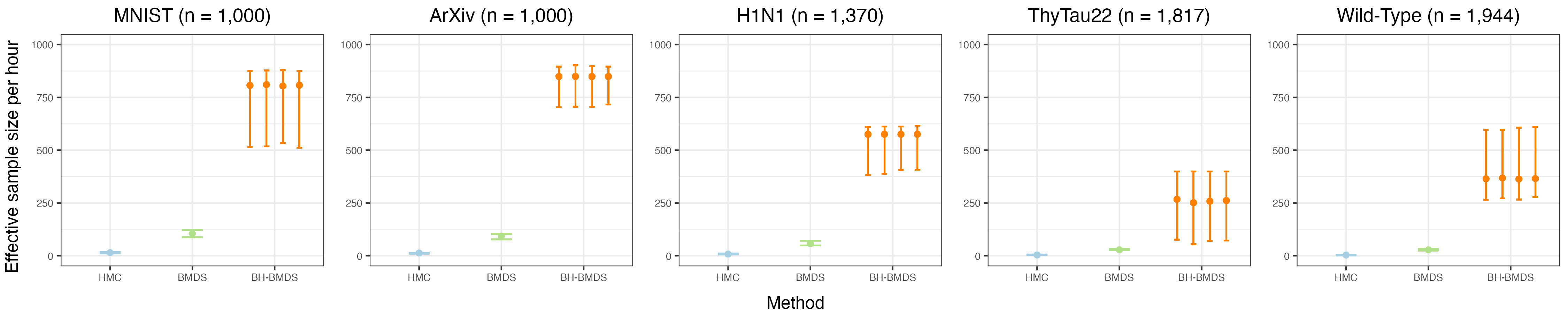}
    \caption{95\% empirical intervals of computational efficiency for H1N1, ArXiv, MNIST and ThyTau22 and Wild-Type cohorts at Week 9 datasets. Computational efficiency is measured as the effective sample size (ESS) per hour, in which ESS is a function of asymptotic auto-correlation. We compare efficiency across three frameworks: 1) BMDS log-likelihood and gradient using HMC (HMC), 2) BMDS log-likelihood using MwG (BMDS) and 3) BH-BMDS log-likelihood using MwG (BH-BMDS, Algorithm \ref{alg:MCMC}). ThyTau22 has a lower ESS per hour compared to Wild-Type for similar number of points due to weaker clustering.}
    \label{fig:eff_small}
\end{figure}

\begin{figure}
    \centering
    \includegraphics[width=\linewidth]{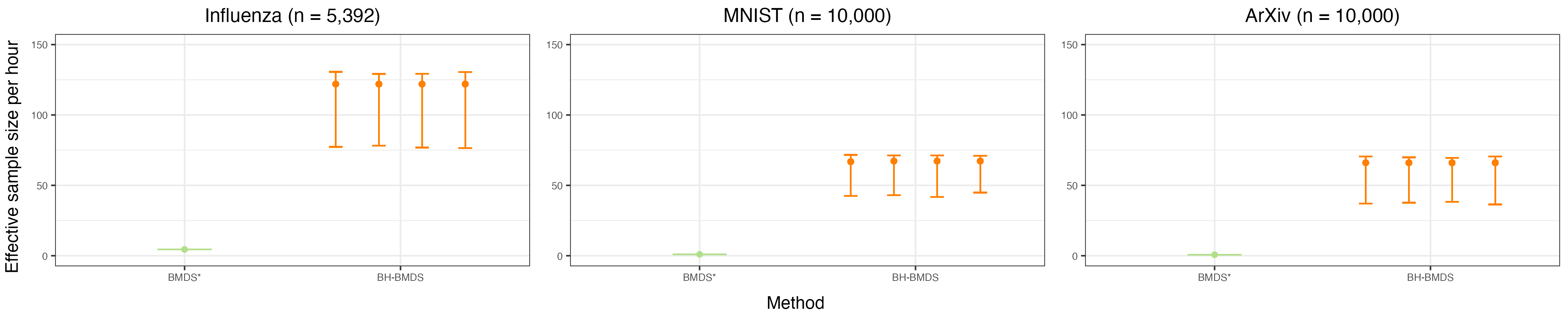}
    \caption{95\% empirical intervals of computational efficiency for influenza, ArXiv and MNIST datasets. Computational efficiency is measured as the effective sample size (ESS) per hour, in which ESS is a function of asymptotic auto-correlation. Since efficiency would be infeasible to compute for such large data under BMDS, we provide a theoretical value. Time at $n = 10{,}000$ is predicted from a quadratic regression using observed times at $n = 100, 500, 1{,}000, 5{,}000$ while ESS is generously assumed to be perfect. BMDS* equals 4.5, 1.0 and 0.7 for influenza, MNIST and ArXiv, respectively.}
    \label{fig:eff_big}
\end{figure}

\section{Discussion} \label{sec:discussion}
We propose a scalable approximation to the BMDS likelihood using Barnes--Hut algorithms (BH-BMDS), reducing computational complexity from $O(n^2)$ to $O(n \log n)$. By embedding this approximation within a Metropolis-within-Gibbs sampler, our approach enables efficient posterior inference while maintaining a probabilistic framework. Crucially, we use a novel approach (Section \ref{app:main_prop}) to establish consistency of the stationary measure for BH-BMDS and evaluate the proposed method on datasets with diverse clustering patterns. 

Several limitations suggest directions for future work. First, our inference scheme relies on coordinate-wise updates of latent locations, which may mix slowly in settings with strong dependence across points. We explored joint updates using HMC, discontinuous HMC and Metropolis--Hastings but encountered difficulties achieving both stable convergence and favorable scaling for larger $n$. More broadly, these findings highlight the challenges of combining discrete approximation schemes with MCMC methods. Second, the noisy Barnes--Hut approximation (Section \ref{sec:stochBH}) introduces stochasticity through the traversal criterion. While this improves mixing by mitigating abrupt changes in the approximated likelihood, it also complicates the theoretical scaling behavior and may deviate from ideal $O(n \log n)$ performance, although such deviations were not observed in practice. Developing more stable approximation schemes that preserve both accuracy and computational guarantees is an important direction for future work. Finally, as with standard BMDS, low-dimensional embeddings may not capture all relevant structure in complex datasets, particularly when restricted to two dimensions. Extending BH-BMDS to higher-dimensional embeddings or adaptive representations is a promising avenue for future work. 

%\subsubsection*{Broader Impact Statement}

%\subsubsection*{Author Contributions}

\subsubsection*{Acknowledgments}
This work was supported by the NSF (DMS 2152774 and DMS 2236854) and NSERC (RGPIN-2022-03012). We also thank Ryan O'Dell, a PhD student at the University of California, Los Angeles, for implementing MaGNet \citep{magnet} and providing the high-dimensional embeddings on the neural activity recordings.

\bibliography{ref}
\bibliographystyle{tmlr}

\appendix
\section{Consistency of approximate BMDS}\label{app:PCproof}
\setcounter{equation}{0}
\counterwithin{equation}{section}
\setcounter{theorem}{0}
\counterwithin{theorem}{section}
\setcounter{corollary}{0}
\counterwithin{corollary}{section}
\setcounter{proposition}{0}
\counterwithin{proposition}{section}
\setcounter{remark}{0}
\counterwithin{remark}{section}
\setcounter{definition}{0}
\counterwithin{definition}{section}
\setcounter{assumption}{0}
\counterwithin{assumption}{section}

We restate and then prove Theorem~\ref{thm:NoisyPC} from the main document. 

We begin by recalling some basic notation from Section~\ref{sec:PostCons}. Throughout Appendix~\ref{app:PCproof}, fix $d\ge 1$, $\sigma>0$, and $M\ge\sqrt d$. We use the shorthand $N=\binom n2$, $\mb J= \mb I-n^{-1}\1\1^\top$, and 
\begin{equation*}
r_{ij}^\star=\norm{\latentdata_i^\star-\latentdata_j^\star},
\qquad
r_{ij}(\latentData)=\norm{\latentdata_i-\latentdata_j}.
\end{equation*}

Our overall goal is to infer the latent positions $\latentdata_{1}^{\star}, \latentdata_{2}^{\star}, \ldots, \latentdata_{n}^{\star} \in [0,1]^{d}$, which we write in matrix form as $\latentData^{\star}$ given the observed dissimilarity matrix $\tilde{D}_{ij} \in [0, M]$. For any matrix $\latentData$ or the actual latent positions $\latentData^{\star}$, define the centered configurations $\latentData_c=\mb J\latentData$ and $\latentData_c^\star=\mb J\latentData^\star$.

Relative to the reference measure
\begin{align}\label{app:eq:refmeasure}
\nu=\delta_0+\lambda|_{(0,M)}+\delta_M
\end{align}
on $[0,M]$, the conditional law of $\tilde D_{ij}$ given $r_{ij}^\star=r$ has density
\begin{align}\label{eq:cappedDensity}
q_{M,r}(y)=
\begin{cases}
\Phi(-r/\sigma), & y=0,\\[2mm]
\dfrac{1}{\sigma}\phi\!\left(\dfrac{y-r}{\sigma}\right), & 0<y<M,\\[2mm]
1-\Phi\!\left(\dfrac{M-r}{\sigma}\right), & y=M,
\end{cases}
\end{align}
where $\lambda$ is the Lebesgue measure and $\phi,\Phi$ denote the standard normal density and CDF, respectively. 

We also recall the following notation from Section~\ref{sec:PostCons}: the loss $\ell_M(y,r)=-\log q_{M,r}(y)$ \Eqref{eq:loss}, the empirical risk $R_n(\latentData)=N^{-1}\sum_{i<j}\ell_M(\tilde D_{ij},r_{ij}(\latentData))$ \Eqref{eq:empRisk}, and for an algorithmic realization $\xi$, the surrogate risk $\widehat R_n^{(\xi)}(\latentData)=N^{-1}\sum_{i<j}\widehat\ell_{ij}^{(\xi)}(\latentData)$ \Eqref{eq:approxRisk}. These yield the exact posterior
\begin{align} \label{app:eq:exactPosterior}
    \Pi_n(d\latentData\mid\mb D)\propto \Pi_0(d\latentData)\exp\{-N R_n(\latentData)\}
\end{align} and the surrogate posterior 
\begin{align} \label{app:eq:approxPosterior}
    \widehat\Pi_n^{(\xi)}(d\latentData\mid\mb D)\propto\Pi_0(d\latentData)\exp\{-N\widehat R_n^{(\xi)}(\latentData)\}.
\end{align}

Since distances are invariant under Euclidean isometries, the posterior distribution cannot possibly be consistent under the Euclidean metric. Instead, we prove consistency in the following "orbit distance" obtained by taking the Frobenius distance $\norm{\mb A}_F=(\sum_{i,j}A_{ij}^2)^{1/2}$ and "modding out" by isometries:

\begin{definition}[Orbital distance]\label{rmk1:OD}
Let $\cG = \{(\mb Q, \mb b): \mb Q\in \cO(d),\, \mb b\in \RR^d\}$ act on configurations via $(\mb Q, \mb b)\cdot \latentdata_i = \mb Q\latentdata_i + \mb b$. Define the orbit distance
\begin{align}\label{app:eq:ProcrustesMetric}
d_{\cG}(\latentData, \latentData^\star)\ =\ \inf_{(\mb Q, \mb b)\in\cG}\ \frac{1}{\sqrt n}\norm{\latentData-(\latentData^\star \mb Q+\1 \mb b^\top)}_F\ =\ \inf_{\mb Q\in\cO(d)}\ \frac{1}{\sqrt n}\norm{\latentData_c-\latentData_c^\star \mb Q}_F.
\end{align}
\end{definition}

Before proving our main result (Theorem \ref{thm:NoisyPC}), we will first establish posterior consistency for the baseline algorithm with parameter $\xi = \xi_0$. For notational simplicity, we drop the algorithmic superscript until Section \ref{SecFinalAppA} and let $\widehat{\ell}_{ij}\equiv\widehat{\ell}_{ij}^{(\xi_0)}$, $\widehat R_n\equiv\widehat R_n^{(\xi_0)}$, and $\widehat{\Pi}_n \equiv \widehat{\Pi}_n^{(\xi_0)}$ denote the surrogate quantities under the baseline.

The key assumption is that the surrogate likelihoods have errors going to 0:

\begin{assumption}[Small Approximation Error]\label{app:A:error}
There exists a deterministic sequence $\varepsilon_n\downarrow 0$ such that, with probability $\to 1$ under the joint law of $(\latentData^\star,\mb D)$,
\begin{align}\label{app:eq:pairwiseError}
\sup_{\latentData\in([0,1]^d)^n}\ \max_{1\le i<j\le n}\ \Big|\widehat\ell_{ij}(\latentData)-\ell_M(\tilde D_{ij},r_{ij}(\latentData))\Big|
\ \le\ \varepsilon_n.
\end{align}
\end{assumption}

Note that, on the high-probability event \Eqref{app:eq:pairwiseError}, the error in the risk is also small:
\begin{align}\label{app:eq:riskError}
\sup_{\latentData\in([0,1]^d)^n}\big|\widehat R_n(\latentData)-R_n(\latentData)\big|\le \varepsilon_n.
\end{align}

We also restate the assumptions about the model:

\begin{assumption}[Nice Data-Generating Process]\label{app:A:model}
The data generating process and prior satisfy:
\begin{enumerate}[label=(T\arabic*)]
\item \textbf{(Random design)} $\latentdata_1^\star,\dots,\latentdata_n^\star \stackrel{\text{i.i.d.}}{\sim}\mathrm{Unif}([0,1]^d)$.
\item \textbf{(Independent pair noise)} Conditional on $\latentData^\star$, the noises $\{Z_{ij}\}_{i<j}$ are i.i.d.\ $N(0,1)$.
\item \textbf{(Prior regularity)} The prior $\Pi_0$ is supported on $([0,1]^d)^n$ and has density $\pi_0$ satisfying $e^{-cn}\le \pi_0(\latentData)\le e^{cn}$ on $([0,1]^d)^n$, for some constant $c>0$.
\end{enumerate}
\end{assumption}

For the remainder of Appendix~\ref{app:PCproof}, all probabilistic statements involving the data-generating law are understood under Assumption~\ref{app:A:model}, unless otherwise stated. Statements about the baseline surrogate posterior are understood in addition under Assumption~\ref{app:A:error} or the high-probability event \Eqref{app:eq:riskError}, as indicated.

\begin{proposition}[Consistency of the Baseline Algorithm]\label{app:main_prop}
Fix $d\ge 1$, $\sigma>0$, and $M\geq \sqrt d$. Consider Model \Eqref{eq:bmds_model} where $r_{ij} = r_{ij}^\star$, the exact posterior \Eqref{app:eq:exactPosterior} based on $\ell_M$ in \eqref{eq:loss}, and the approximate posterior $\widehat{\Pi}_n$ \Eqref{app:eq:approxPosterior}. Suppose Assumptions~\ref{app:A:error} and \ref{app:A:model} hold. Then, for every $\delta>0$,
\begin{align}
    \widehat \Pi_n\Big(\big\{\latentData:\ d_{\cG}(\latentData,\latentData^\star)>\delta\big\}\ \Big|\ \mb D\Big)\ \xrightarrow{\Pbb}\ 0.
\end{align}
In particular, the approximate posterior of the baseline algorithm is consistent for the latent configuration up to rigid motion.
\end{proposition}

We restate Theorem~\ref{thm:NoisyPC} slightly more formally:

\begin{theorem}[Theorem~\ref{thm:NoisyPC}, Extended]\label{thm:NoisyPC_ext}
Fix $d\ge 1$, $\sigma>0$, $M\geq \sqrt d$ and a parameter space $\Xi$. Suppose Assumption~\ref{app:A:model} holds, and suppose there exists a reference parameter $\xi_0 \in \Xi$ such that the chain $P_{\xi_0}$ on $([0,1]^d)^n$ is uniformly ergodic and $\widehat{\ell}_{ij}^{(\xi_0)}$ satisfies Assumption~\ref{app:A:error}. For each $\xi\in\Xi$, let $P_{\xi}$ be a Markov chain on $([0,1]^d)^n$ with stationary measure $\widehat{\Pi}_{n}^{(\xi)}$, and define the mixture kernel $P=\int_{\Xi} P_{\xi}\,\mu(d\xi)$ for a probability measure $\mu$ on $\Xi$. Let $\widetilde\Pi_n$ be an invariant probability measure of $P$. 

Suppose that there exist deterministic constants $C_n<\infty$ and $\rho_n\in(0,1)$ such that
\begin{align}
\sup_{\latentData_0\in([0,1]^d)^n}\big\|\delta_{\latentData_0}P_{\xi_0}^m-\widehat\Pi_n^{(\xi_0)}\big\|_{\mathrm{TV}}
\le C_n\rho_n^m,
\qquad m\ge 0.
\end{align}
Set
\begin{align} \label{app:thm_lkdef}
\lambda_n=\max\left\{0,\left\lceil\frac{\log(1/C_n)}{\log(\rho_n)}\right\rceil\right\},
\qquad
\kappa_n=\lambda_n+\frac{C_n\rho_n^{\lambda_n}}{1-\rho_n},
\end{align}
and assume that the possibly data-dependent quantity
\begin{align}\label{app:thm_deltadef}
\Delta_n =
\sup_{\xi\in\Xi}\ \sup_{\latentData_0\in([0,1]^d)^n}
\big\|P_{\xi}(\latentData_0,\cdot)-P_{\xi_0}(\latentData_0,\cdot)\big\|_{\mathrm{TV}}
\end{align}
satisfies $\kappa_n\Delta_n\xrightarrow{\Pbb}0$.

Then for every $\delta>0$,
\begin{align}
\widetilde\Pi_n\Big(\big\{\latentData:\ d_{\cG}(\latentData,\latentData^\star)>\delta\big\}\Big)\xrightarrow{\Pbb}0.
\end{align}
\end{theorem}

\subsection{Proof strategy for Proposition~\ref{app:main_prop}}

The basic idea is: 

\begin{enumerate}
    \item Check that the empirical risk $R_{n}$ concentrates near an associated "population" risk $R_{n}^{0}$ (Section \ref{app:SubsecUnifConc}),
    \item Check that the population risk $R_{n}^{0}(\latentData)$ grows quadratically in the \textit{orbit} (not Euclidean) distance to the truth $\latentData^{\star}$, with reasonable constant (Section \ref{app:SubsecQuadraticMargin}),
    \item Check that the posterior associated with the population risk is consistent (Section \ref{app:SecConsistencyPost}).
\end{enumerate}

We note that this approach is not the most common approach to analyzing the error of "approximate" or "surrogate" MCMC algorithms. Typically, one aims to have a surrogate that is pointwise close to the target log-likelihood, then use perturbation theory to show that the stationary measure is close to the posterior distribution of interest (e.g., \cite{RastelliEtAl} for a typical analysis, applied to a different MCMC algorithm for inferring latent positions). In the present paper, we allow the surrogate to be very far from the log-likelihood pointwise, as long as the error grows only sub-quadratically in the number of points $n$. 

We are aware of several other results that prove consistency of MCMC algorithms even as the pointwise error is allowed to grow, though only in the much simpler situation of i.i.d.~samples and a finite-dimensional set of parameters to be learned. The influential analysis of SGLD in \cite{TehThieryVollmer2016} is an early example, and this viewpoint is still being refined through works such as \cite{bieringer2026surrogate}.

We emphasize the following differences between our analysis and most or all previous work:

\begin{enumerate}
    \item Directly analyzing the risk is always more difficult than merely treating an approximate MCMC chain as a perturbation of an ideal chain, where more generic results are available \citep{Mitrophanov2005,AlquierFrielEverittBoland2016,MedinaAguayoLeeRoberts2016,RudolfSchweizer2018}.
    \item The current problem is effectively infinite-dimensional, since the number of latent positions to learn grows with the size of the dataset. For this reason, naive bounds based on, essentially, Taylor expansions and concentration do not work. Compare this to the more classical finite-dimensional analyses of stochastic-gradient MCMC in \cite{TehThieryVollmer2016}.
    \item The ground truth is only identifiable up to an isometry, so all estimates need to effectively "mod out" by this isometry.
\end{enumerate}

\subsection{A standard matrix inequality}

\begin{lemma} \label{app:LemmaOpNorm}
Define the linear "double-centering" operator $\mathcal{C}(\mb A)=\mb J \mb A \mb J$ on $\RR^{n\times n}$, where
$\mb J= \mb I-\frac1n \1\1^\top$ is the orthogonal projector onto $\1^\perp$. We have the operator norm bound
\begin{equation*}
\|\mathcal{C}\|_{F\to F}\le 1.
\end{equation*}
\end{lemma}

\begin{proof}
Note $\mb J^\top=\mb J$, $\mb J^2= \mb J$ and
$\|\mb J\|_2=1$. Using the standard inequality $\|\mb U \mb A \mb V\|_F\le \|\mb U\|_2\|\mb A\|_F\|\mb V\|_2$
(e.g., the chapter on matrix norms \cite{HornJohnson2013}),
we obtain for all $\mb A$ that
\begin{equation*}
\|\mb J \mb A\mb J\|_F \le \|\mb J\|_2^2 \|\mb A\|_F = \|\mb A\|_F,
\end{equation*}
so $\|\mathcal{C}\|_{F\to F}\le 1$.
\end{proof}

\subsection{Uniform concentration of the empirical risk} \label{app:SubsecUnifConc}

Define the conditional (given $\latentData^\star$) population risk
\begin{equation}\label{app:eq:popRisk}
R_n^0(\latentData)\ =\ \E\!\left[R_n(\latentData)\mid \latentData^\star\right]
\;=\;\frac1N\sum_{i<j} g\big(r_{ij}(\latentData); r_{ij}^\star\big),
\end{equation}
where, for $Z\sim N(0,1)$,
\begin{equation}\label{app:eq:popRiskPt}
g(r;r^\star)=\E\Big[\ell_M\big(T_M([r^\star+\sigma Z]_+),r\big)\Big].
\end{equation}

\begin{lemma}[Uniform  concentration]\label{app:lem:ULLN}
There exist constants $C_1=C_1(d,\sigma,M)>0$, $C_2=C_2(d,\sigma,M)>0$, and $c=c(d,\sigma,M)>0$ such that, for all $0<t<C_1$,
\begin{equation*}
\Pbb\Big(\ \sup_{\latentData\in([0,1]^d)^n}\big|R_n(\latentData)-R_n^0(\latentData)\big|\ >\ t\ \Big|\ \latentData^\star\Big)
\ \le\ 2\exp\Big(C_2\,nd\log\left(\frac{C_1}{t}\right)\ -\ c\,N t^2\Big).
\end{equation*}
\end{lemma}

\begin{proof}
Fix $\latentData$. Conditional on $\latentData^\star$, the random variables
\begin{equation*}
Y_{ij}(\latentData)=\ell_M(\tilde D_{ij},r_{ij}(\latentData))
\end{equation*}
are independent across pairs $(i,j)$.

On $0<y<M$ we have $|y-r|\le M$, while at the atoms
\begin{equation*}
q_{M,r}(0)=\Phi(-r/\sigma)\ge \Phi(-\sqrt d/\sigma),
\qquad
q_{M,r}(M)=1-\Phi((M-r)/\sigma)\ge \Phi(-M/\sigma).
\end{equation*}
Thus, because $r\in[0,\sqrt d]$ and $y\in[0,M]$, the density $q_{M,r}(y)$ is uniformly bounded above and below by constants $0<q_-\le q_+<\infty$ depending only on $(d,\sigma,M)$. Hence
\begin{equation}\label{app:eq:Yij_bounded}
Y_{ij}(\latentData)\in[-\log q_+,-\log q_-]
\end{equation}
almost surely. By \eqref{app:eq:Yij_bounded}, writing $W_\ell=\log(q_+/q_-)$, Hoeffding's inequality gives
\begin{equation} \label{app:IneqHoeffStep1}
\Pbb\Big(\big|R_n(\latentData)-R_n^0(\latentData)\big|>t\ \Big|\ \latentData^\star\Big)
\ \le\ 2\exp\Big(-\frac{2Nt^2}{W_\ell^2}\Big).
\end{equation}
To make the bound uniform in $\latentData$, we use a net argument on $([0,1]^d)^n$.
Equip this space with the sup-norm on rows metric $\norm{\latentData-\latentData'}_\infty=\max_i\norm{\latentdata_i-\latentdata_i'}$.
It is well-known \citep[Corollary~4.2.11]{VershyninHDP2} that there exists a universal $0 < C_{0} < \infty$ so that, for any fixed $\eta > 0$, there exists an $\eta$-net $\cN_\eta$ satisfying
\begin{equation} \label{app:IneqNetSize}
|\cN_\eta|\ \le\ \Big(\frac{C_0}{\eta}\Big)^{nd}.
\end{equation}
Next, note that since $\latentData, \latentData' \in ([0,1]^d)^n$, we have that $r_{ij}(\latentData)$ is $2$-Lipschitz in $\norm{\cdot}_\infty$. Also, for each fixed $y\in[0,M]$, the map $r\mapsto \ell_M(y,r)$ is $C^1$ on $[0,\sqrt d]$, with derivative
\begin{equation*}
\frac{\partial}{\partial r} \ell_M(y,r)=
\begin{cases}
\dfrac{\phi(r/\sigma)}{\sigma\Phi(-r/\sigma)}, & y=0,\\[2mm]
\dfrac{r-y}{\sigma^2}, & 0<y<M,\\[2mm]
-\dfrac{\phi((M-r)/\sigma)}{\sigma(1-\Phi((M-r)/\sigma))}, & y=M.
\end{cases}
\end{equation*}
Each of these three branches is uniformly bounded on the relevant compact interval, so there exists $L_\ell=L_\ell(d,\sigma,M)<\infty$ such that
\begin{equation*}
|\ell_M(y,r)-\ell_M(y,r')|\le L_\ell |r-r'|,
\qquad y\in[0,M],\ r,r'\in[0,\sqrt d].
\end{equation*}
Combining this with the fact that $r_{ij}(\latentData)$ is $2$-Lipschitz in $\norm{\cdot}_\infty$, we have
\begin{equation}\label{app:eq:pair_loss_lipschitz}
|\ell_M(\tilde D_{ij},r_{ij}(\latentData))-\ell_M(\tilde D_{ij},r_{ij}(\latentData'))|
\le 2L_\ell\norm{\latentData-\latentData'}_\infty.
\end{equation}
Averaging over pairs gives
\begin{equation}\label{app:eq:emp_risk_lipschitz}
|R_n(\latentData)-R_n(\latentData')|\le 2L_\ell\norm{\latentData-\latentData'}_\infty.
\end{equation}
Since $R_{n}^{0}(\latentData)$ is a conditional expectation of $R_{n}(\latentData)$, we have the same bound:
\begin{equation}\label{app:eq:pop_risk_lipschitz}
|R_n^0(\latentData)-R_n^0(\latentData')|\le 2L_\ell\norm{\latentData-\latentData'}_\infty.
\end{equation}

Pick $\eta=t/(8L_\ell)$. For any $\latentData$, choose $\latentData'\in\cN_\eta$ with $\norm{\latentData-\latentData'}_\infty\le \eta$. By the triangle inequality, \eqref{app:eq:emp_risk_lipschitz} and \eqref{app:eq:pop_risk_lipschitz},
\begin{align}
\label{app:eq:net_reduction}
|R_n(\latentData)-R_n^0(\latentData)|
&\le |R_n(\latentData')-R_n^0(\latentData')| + |R_n(\latentData)-R_n(\latentData')| + |R_n^0(\latentData)-R_n^0(\latentData')| \notag\\
&\le |R_n(\latentData')-R_n^0(\latentData')| + 4L_\ell\eta
= |R_n(\latentData')-R_n^0(\latentData')| + \frac{t}{2}.
\end{align}
Therefore, by \eqref{app:eq:net_reduction}, a union bound over $\cN_\eta$ together with \eqref{app:IneqHoeffStep1} applied at level $t/2$ gives
\begin{equation*}
\Pbb\Big(\sup_{\latentData\in([0,1]^d)^n}|R_n(\latentData)-R_n^0(\latentData)|>t\ \Big|\ \latentData^\star\Big)
\le
2|\cN_\eta|\exp\Big(-\frac{Nt^2}{2W_\ell^2}\Big).
\end{equation*}
Substituting \eqref{app:IneqNetSize} and $\eta=t/(8L_\ell)$ and then enlarging constants gives the conclusion for all sufficiently small $t$.
\end{proof}

\subsection{A quadratic lower bound for the population risk} \label{app:SubsecQuadraticMargin}

We show $R_n^0(\latentData)$ is minimized (up to isometry) at $\latentData^\star$, and furthermore that the risk increases quadratically as you move away from $\latentData^\star$. More precisely, we will show

\begin{proposition}\label{app:prop:marginOrbit}
There exists $c_{\mathrm{mg}}=c_{\mathrm{mg}}(d,\sigma,M)>0$ such that the following holds (uniformly in $\latentData$) with probability going to 1 as $n$ goes to infinity: 
\begin{equation} \label{app:IneqStep2Conclusion}
R_n^0(\latentData)-R_n^0(\latentData^\star)\ \ge\ c_{\mathrm{mg}}\ d_{\cG}(\latentData,\latentData^\star)^2,
\qquad \forall \latentData\in([0,1]^d)^n.
\end{equation}
\end{proposition}

This is the main result in this section. The proof is deferred to the end, in sub-section \ref{app:SecQuadMarginProof}.

\subsubsection{Initial quadratic lower bound in the distance}

\begin{lemma}\label{app:lem:strongConvex}
There exists $c_{\mathrm{sc}}=c_{\mathrm{sc}}(\sigma,M,d)>0$ such that for all $r,r^\star\in[0,\sqrt{d}]$,
\begin{equation*}
g(r;r^\star)-g(r^\star;r^\star)\ \ge\ c_{\mathrm{sc}}\,(r-r^\star)^2,
\end{equation*}
where $g$ is defined in \eqref{app:eq:popRiskPt}.
\end{lemma}

\begin{proof}
For fixed $r^\star$, let $P_r$ denote the law of $\tilde D=T_M([r+\sigma Z]_+)$, with density $q_{M,r}$ with respect to $\nu$; equivalently, $P_r(A)=\int_A q_{M,r}(y)\,\nu(dy)$ for measurable $A\subseteq[0,M]$. Thus the observation generated at the true distance $r^\star$ has law $P_{r^\star}$, while evaluating the loss at a candidate distance $r$ uses the density $q_{M,r}$. Therefore, by the definition of $g$ in \eqref{app:eq:popRiskPt},
\begin{equation}\label{app:eq:scalar_KL_identity}
g(r;r^\star)-g(r^\star;r^\star)
=
\int q_{M,r^\star}(y)\log\!\left(\frac{q_{M,r^\star}(y)}{q_{M,r}(y)}\right)\,\nu(dy)
=
\mathrm{KL}(P_{r^\star}\,\|\,P_r).
\end{equation}
By Pinsker's inequality and \eqref{app:eq:scalar_KL_identity},
\begin{equation}\label{app:eq:pinsker_scalar}
g(r;r^\star)-g(r^\star;r^\star)
\ge 2\|P_{r^\star}-P_r\|_{\mathrm{TV}}^2.
\end{equation}
Recalling the definition of $\mu_{M}$ from \eqref{eq:noisyD}, we have
\begin{align}
\label{app:eq:mean_tv_bound}
|\mu_M(r)-\mu_M(r^\star)|
&=
\left|\int_0^M y\,\big(q_{M,r}(y)-q_{M,r^\star}(y)\big)\,\nu(dy)\right| \notag\\
&\le
M\int_0^M \big|q_{M,r}(y)-q_{M,r^\star}(y)\big|\,\nu(dy) \notag\\
&=
2M\|P_r-P_{r^\star}\|_{\mathrm{TV}}.
\end{align}
Combining \eqref{app:eq:pinsker_scalar} and \eqref{app:eq:mean_tv_bound} gives
\begin{equation}\label{app:eq:kl_mu_lower}
g(r;r^\star)-g(r^\star;r^\star)
\ge \frac{\big(\mu_M(r)-\mu_M(r^\star)\big)^2}{2M^2}.
\end{equation}
For each fixed $Z$, the map $r\mapsto T_M([r+\sigma Z]_+)$ is piecewise linear with derivative $\1\{0<r+\sigma Z<M\}$ for almost every $r$, so dominated convergence yields
\begin{equation}\label{app:eq:mu_derivative}
\mu_M'(r)=\Pbb(0<r+\sigma Z<M)=\Phi\!\left(\frac{M-r}{\sigma}\right)-\Phi\!\left(-\frac{r}{\sigma}\right).
\end{equation}
Hence, for $r\in[0,\sqrt d]$,
\begin{equation}\label{app:eq:mu_derivative_lower}
\mu_M'(r)=\Phi\!\left(\frac{M-r}{\sigma}\right)+\Phi\!\left(\frac{r}{\sigma}\right)-1
\ge \Phi\!\left(\frac{M}{2\sigma}\right)-\frac12=:m_M.
\end{equation}
We note that $m_{M} > 0$, since $r\in[0,M]$ and therefore at least one of $r$ and $M-r$ is at least $M/2$. By the mean value theorem and \eqref{app:eq:mu_derivative_lower},
\begin{equation}\label{app:eq:mu_mvt_lower}
|\mu_M(r)-\mu_M(r^\star)|\ge m_M|r-r^\star|.
\end{equation}
Substituting \eqref{app:eq:mu_mvt_lower} into \eqref{app:eq:kl_mu_lower} proves the claim, with $c_{\mathrm{sc}}=m_M^2/(2M^2)$.
\end{proof}

\subsubsection{From pairwise distance error to Procrustes error}

Define the average squared distance discrepancy
\begin{equation*}
\Delta_{\mathrm{dist}}(\latentData,\latentData^\star)=\frac1N\sum_{i<j}\big(r_{ij}(\latentData)-r_{ij}^\star\big)^2.
\end{equation*}

We have the straightforward estimate:

\begin{lemma}[Population risk dominates distance discrepancy]\label{app:lem:riskToDist}
Under the model~\ref{eq:bmds_model} and loss \Eqref{eq:loss},
\begin{equation*}
R_n^0(\latentData)-R_n^0(\latentData^\star)\ \ge\ c_{\mathrm{sc}}\ \Delta_{\mathrm{dist}}(\latentData,\latentData^\star),
\end{equation*}
where $c_{\mathrm{sc}}$ is as in Lemma~\ref{app:lem:strongConvex}.
\end{lemma}

\begin{proof}
By Lemma~\ref{app:lem:strongConvex}, for each pair $(i,j)$,
\begin{equation*}
g\big(r_{ij}(\latentData);r_{ij}^\star\big)-g\big(r_{ij}^\star;r_{ij}^\star\big)\ \ge\ c_{\mathrm{sc}}\big(r_{ij}(\latentData)-r_{ij}^\star\big)^2.
\end{equation*}
Averaging over pairs yields the result.
\end{proof}

Next we relate $\Delta_{\mathrm{dist}}(\latentData,\latentData^\star)$ to $d_{\cG}(\latentData,\latentData^\star)$ using a standard Procrustes perturbation bound. Recall that $\latentData_c=\mb J\latentData$ and $\latentData_c^\star= \mb J\latentData^\star$ denote the centered configurations.

We will use the following consequence of Theorem~1 \cite{AriasCastroJavanmardPelletier2020} as our basic Procrustes perturbation bound: with $\varepsilon^2=\norm{\mb Y \mb Y^\top-\mb X \mb X^\top}_F$, if $\norm{\mb X^\dagger}_{\mathrm{op}}\varepsilon\le 1/\sqrt2$, then
\begin{equation*}
\min_{\mb Q\in\cO(d)}\norm{\mb Y- \mb X \mb Q}_F
\le (1+\sqrt2)\norm{\mb X^\dagger}_{\mathrm{op}}\norm{\mb Y \mb Y^\top- \mb X \mb X^\top}_F.
\end{equation*}
Note that the theorem in that paper looks slightly different and is quite a bit more complicated; this follows immediately from the second branch in \cite[Theorem~1, inequality~(5)]{AriasCastroJavanmardPelletier2020} after choosing Schatten norm $p=2$, simplifying and rephrasing in our notation.

\begin{lemma}\label{app:lem:distToProc}
For every $\alpha>0$, there exists $c_2=c_2(\alpha,d)>0$ such that, if
\begin{equation}\label{app:eq:wellconditioned}
s_{\min}(\latentData_c^\star)\ \ge\ \alpha\sqrt{n},
\end{equation}
then
\begin{equation*}
d_{\cG}(\latentData,\latentData^\star)^2\ \le\ c_2\ \Delta_{\mathrm{dist}}(\latentData,\latentData^\star),
\qquad \forall \latentData\in([0,1]^d)^n.
\end{equation*}
\end{lemma}

\begin{proof}
We follow  \cite{AriasCastroJavanmardPelletier2020} and work with centered Gram matrices.
Denote by $D^2(\latentData)$ the squared-distance matrix with entries $D^{2}(\latentData)_{ij} = \|\latentdata_i -\latentdata_j\|^{2}$. With $\mb J$ as above, the centered Gram matrix is
\begin{equation*}
G(\latentData)=-\frac12\, \mb J D^2(\latentData) \mb J = \latentData_c\latentData_c^\top.
\end{equation*}

We begin by bounding $G(\latentData)-G(\latentData^\star)$ in terms of $\Delta_{\mathrm{dist}}(\latentData,\latentData^\star)$. By Lemma~\ref{app:LemmaOpNorm}, the Frobenius--Frobenius operator norm of $\mb A\mapsto \mb J\mb A\mb J$ is $1$, so
\begin{equation}\label{app:eq:IneqGD_new}
\norm{G(\latentData)-G(\latentData^\star)}_F
= \frac12\norm{\mb J\big(D^2(\latentData)-D^2(\latentData^\star)\big)\mb J}_F
\le \frac12\,\norm{D^2(\latentData)-D^2(\latentData^\star)}_{F}.
\end{equation}

For $r,r^\star\in[0,\sqrt d]$ we have
\begin{equation}\label{app:eq:square_to_linear_new}
|r^2-(r^\star)^2|
= |r-r^\star|\,(r+r^\star)
\le 2\sqrt d\,|r-r^\star|.
\end{equation}
Therefore, by \eqref{app:eq:square_to_linear_new},
\begin{align}
\label{app:eq:IneqDDelta_new}
\norm{D^2(\latentData)-D^2(\latentData^\star)}_F^{2}
&= \sum_{i\ne j}\bigl(r_{ij}(\latentData)^{2} - (r_{ij}^{\star})^{2}\bigr)^{2} \notag\\
&\le 4d \sum_{i\ne j}\bigl(r_{ij}(\latentData) - r_{ij}^{\star}\bigr)^{2}
\\
&= 8d\sum_{i<j}\bigl(r_{ij}(\latentData)-r_{ij}^{\star}\bigr)^{2} \notag\\
&= 8dN \,\Delta_{\mathrm{dist}}(\latentData,\latentData^\star).
\end{align}

Combining \eqref{app:eq:IneqGD_new} and \eqref{app:eq:IneqDDelta_new} gives
\begin{equation}\label{app:eq:Gram_from_Delta_new}
\norm{G(\latentData)-G(\latentData^\star)}_F \le \sqrt{2dN\,\Delta_{\mathrm{dist}}(\latentData,\latentData^\star)}.
\end{equation}

Let
\begin{equation}\label{app:eq:ACJP_eps_def_new}
\varepsilon = \norm{G(\latentData)-G(\latentData^\star)}_F^{1/2}.
\end{equation}

We now apply \cite[Theorem~1]{AriasCastroJavanmardPelletier2020} with $\mb X=\latentData_c^\star$ and $\mb Y=\latentData_c$ and Schatten norm $p=2$
(which is equal to the Frobenius norm). In the notation of that theorem, we have
$\mb Y\mb Y^\top-\mb X\mb X^\top = G(\latentData)-G(\latentData^\star)$ and hence $\varepsilon^2=\|\mb Y\mb Y^\top-\mb X\mb X^\top\|_F$.
Using the second branch in the statement of \cite[Theorem~1, inequality (5)]{AriasCastroJavanmardPelletier2020}, if
\begin{equation}\label{app:eq:ACJP_condition_new}
\|(\latentData_c^\star)^\dagger\|_{\mathrm{op}}\,\varepsilon \le \frac1{\sqrt2},
\end{equation}
then
\begin{equation}\label{app:eq:ACJP_bound_new}
\min_{\mb Q\in\cO(d)}\norm{\latentData_c - \latentData_c^\star \mb Q}_F
\le (1+\sqrt2)\,\|(\latentData_c^\star)^\dagger\|_{\mathrm{op}}\ \norm{G(\latentData)-G(\latentData^\star)}_F.
\end{equation}

We now give a sufficient condition for \eqref{app:eq:ACJP_condition_new}. Under \eqref{app:eq:wellconditioned}, we have $\|(\latentData_c^\star)^\dagger\|_{\mathrm{op}} = 1/s_{\min}(\latentData_c^\star)\le 1/(\alpha\sqrt n)$.
Moreover, since $N=\binom{n}{2}\le n^2/2$, Equations~\ref{app:eq:Gram_from_Delta_new}--\ref{app:eq:ACJP_eps_def_new} imply
\begin{equation}\label{app:eq:eps_bound_from_delta}
\varepsilon
\le \bigl(2dN\,\Delta_{\mathrm{dist}}(\latentData,\latentData^\star)\bigr)^{1/4}
\le d^{1/4}n^{1/2}\,\Delta_{\mathrm{dist}}(\latentData,\latentData^\star)^{1/4}.
\end{equation}
Hence, by \eqref{app:eq:wellconditioned} and \eqref{app:eq:eps_bound_from_delta}, \eqref{app:eq:ACJP_condition_new} holds whenever
\begin{equation}\label{app:eq:Delta_small_regime_new}
\Delta_{\mathrm{dist}}(\latentData,\latentData^\star)\ \le\ \frac{\alpha^4}{4d}.
\end{equation}

We now obtain our bound in two cases: when  \eqref{app:eq:Delta_small_regime_new} holds, and when it fails. We begin by assuming it holds. Combining
\eqref{app:eq:ACJP_bound_new} with \eqref{app:eq:Gram_from_Delta_new} and \eqref{app:eq:wellconditioned} gives
\begin{align}
\frac1{\sqrt n}\min_{\mb Q\in\cO(d)}\norm{\latentData_c - \latentData_c^\star \mb Q}_F
&\le \frac{1+\sqrt2}{\alpha n}\ \sqrt{2dN\,\Delta_{\mathrm{dist}}(\latentData,\latentData^\star)} \notag\\
&\le \frac{1+\sqrt2}{\alpha}\ \sqrt{d\,\Delta_{\mathrm{dist}}(\latentData,\latentData^\star)}.
\label{app:eq:rms_bound_small_new}
\end{align}
By the definition of $d_{\cG}$ in \eqref{app:eq:ProcrustesMetric}, the left-hand side equals $d_{\cG}(\latentData,\latentData^\star)$, so
squaring \eqref{app:eq:rms_bound_small_new} yields
\begin{equation}\label{app:eq:final_small_new}
d_{\cG}(\latentData,\latentData^\star)^2
\le \frac{(1+\sqrt2)^2}{\alpha^2}\ d\ \Delta_{\mathrm{dist}}(\latentData,\latentData^\star)
\qquad\text{whenever \eqref{app:eq:Delta_small_regime_new} holds.}
\end{equation}

Finally, we consider the case that  \eqref{app:eq:Delta_small_regime_new} fails. In this case, $\Delta_{\mathrm{dist}}(\latentData,\latentData^\star)\ge \alpha^4/(4d)$.
Since $\latentdata_i,\latentdata_i^\star\in[0,1]^d$, choosing $\mb Q= \mb I$ and $\mb b=\mb 0$ in the uncentered definition of the orbit distance gives the trivial bound $d_{\cG}(\latentData,\latentData^\star)^2\le d$.
Therefore,
\begin{equation}\label{app:eq:final_large_new}
d_{\cG}(\latentData,\latentData^\star)^2\le d
\le \frac{4d^2}{\alpha^4}\ \Delta_{\mathrm{dist}}(\latentData,\latentData^\star)
\qquad\text{whenever \eqref{app:eq:Delta_small_regime_new} fails.}
\end{equation}

Taking
\begin{equation*}
c_2=\max\Big\{\frac{(1+\sqrt2)^2}{\alpha^2}d,\ \frac{4d^2}{\alpha^4}\Big\}
\end{equation*}
and combining \eqref{app:eq:final_small_new} and \eqref{app:eq:final_large_new} proves the lemma.
\end{proof}

\subsubsection{Well-conditioning holds with high probability for uniform points}

\begin{lemma}\label{app:lem:conditioning}
Let $\latentdata_i^\star\stackrel{iid}{\sim}\mathrm{Unif}([0,1]^d)$, let $\latentData^{\star}$ be the associated matrix of points, and let $\latentData_c^\star$ be its centered version. Then there exists $c_1=c_1(d)>0$ such that
\begin{equation*}
\lim_{n \rightarrow \infty} \Pbb\Big(s_{\min}(\latentData_c^\star)\ge c_1\sqrt n\Big) = 1.
\end{equation*}
\end{lemma}

\begin{proof}
Let $\mb u_i=\latentdata_i^\star-\E[\latentdata_i^\star]=\latentdata_i^\star-\frac12\1_d$, and let $\mb U$ be the $n\times d$ matrix with rows $\mb u_i^\top$. Then the rows of $\mb U$ are i.i.d., mean-zero, bounded random vectors in $\RR^d$, and
\begin{equation*}
\latentData_c^\star=\mb J \mb U= \mb U-\1 \bar{\mb u}^\top,
\qquad
\bar{\mb u}=\frac1n\sum_{i=1}^n \mb u_i.
\end{equation*}
We first lower-bound $s_{\min}(\mb U)$. Set $\mb A=\mb U\mb \Sigma^{-1/2}=\sqrt{12}\,\mb U$, where
\begin{equation*}
\mb \Sigma=\E[\mb u_i \mb u_i^\top]=\frac1{12}\mb I_d.
\end{equation*}
Then the rows of $\mb A$ are isotropic, subgaussian and have subgaussian norm bounded by a constant $K_d<\infty$ depending only on $d$. By \cite[Theorem~4.6.1]{VershyninHDP2}, there exists an absolute constant $0 < C < \infty$ so that for every $t>0$,
\begin{equation}\label{app:eq:vershynin_smin}
\Pbb\big[s_{\min}(\mb A)\ge \sqrt n-C K_d^2(\sqrt d+t)\big]\ge 1-2e^{-t^2}.
\end{equation}

Since $\mb A=\sqrt{12}\, \mb U$, choosing $t=\sqrt n/(4CK_d^2)$ in \eqref{app:eq:vershynin_smin} shows that there exist constants $c_U=c_U(d)>0$ and $c_A=c_A(d)>0$ such that, for all sufficiently large $n$,
\begin{equation}\label{app:eq:smin_U_bound}
\Pbb\big(s_{\min}(\mb U)\ge c_U\sqrt n\big) \ge 1-2e^{-c_A n}.
\end{equation}

We next control the rank-one correction $\1\bar{\mb u}^\top$. Because each coordinate of $\mb u_i$ is bounded and mean-zero, Hoeffding's inequality gives, for every $a>0$,
\begin{equation*}
\Pbb\big(\sqrt n\,|\bar u_k|>a\big)\le 2e^{-c a^2},
\qquad k=1,\dots,d,
\end{equation*}
for a constant $c=c(d)>0$. Therefore, by a union bound, for every $b>0$,
\begin{equation}\label{app:eq:ubar_bound}
\Pbb\big(\sqrt n\,\|\bar{\mb u}\|>b\big)
\le 2d\exp\left(-\frac{c b^2}{d}\right).
\end{equation}
Since $\|\1\bar{\mb u}^\top\|_{\mathrm{op}}=\sqrt n\,\|\bar{\mb u}\|$, substituting $b=(c_U/2)\sqrt n$ in \eqref{app:eq:ubar_bound} gives
\begin{equation}\label{app:eq:rank_one_correction_bound}
\Pbb\left(\|\1\bar{\mb u}^\top\|_{\mathrm{op}}> \frac{c_U}{2}\sqrt n\right)
\le 2d\exp\left(-\frac{c c_U^2}{4d}n\right).
\end{equation}
Let
\[
E_n=\left\{s_{\min}(\mb U)\ge c_U\sqrt n\right\}\cap
\left\{\|\1\bar{\mb u}^\top\|_{\mathrm{op}}\le \frac{c_U}{2}\sqrt n\right\}.
\]
By \eqref{app:eq:smin_U_bound} and \eqref{app:eq:rank_one_correction_bound}, $\Pbb(E_n)\to1$. On $E_n$, since $\latentData_c^\star= \mb U-\1\bar{\mb u}^\top$, Weyl's inequality for singular values gives
\begin{equation*}
s_{\min}(\latentData_c^\star)
=s_{\min}(\mb U-\1\bar{\mb u}^\top)
\ge s_{\min}(\mb U)-\|\1\bar{\mb u}^\top\|_{\mathrm{op}}
\ge (c_U/2)\sqrt n.
\end{equation*}
Taking $c_1=c_U/2$ proves the lemma.
\end{proof}

\subsubsection{Completing the quadratic lower bounds} \label{app:SecQuadMarginProof}

\begin{proof}[Proof of Proposition~\ref{app:prop:marginOrbit}]

Let $c_1=c_1(d)$ be the constant from Lemma~\ref{app:lem:conditioning}. We prove \eqref{app:IneqStep2Conclusion} on the high-probability event $s_{\min}(\latentData_c^\star)\ge c_1\sqrt n$, applying Lemma~\ref{app:lem:distToProc} with $\alpha=c_1$.

Combine Lemma~\ref{app:lem:riskToDist} and Lemma~\ref{app:lem:distToProc}:
\begin{equation*}
R_n^0(\latentData)-R_n^0(\latentData^\star)\ \ge\ c_{\mathrm{sc}}\,\Delta_{\mathrm{dist}}(\latentData,\latentData^\star)\ \ge\ c_{\mathrm{sc}}\cdot \frac{1}{c_2}\ d_{\cG}(\latentData,\latentData^\star)^2.
\end{equation*}
Set $c_{\mathrm{mg}}=c_{\mathrm{sc}}/c_2$.

\end{proof}

\begin{proposition}\label{app:prop:upperMarginOrbit}
There exists a constant $C_{\rm up}=C_{\rm up}(d,\sigma,M)>0$ such that, for every $\latentData\in([0,1]^d)^n$,
\begin{equation}\label{app:eq:global_upper_margin}
R_n^0(\latentData)-R_n^0(\latentData^\star)\ \le\ C_{\rm up}\, d_{\cG}(\latentData,\latentData^\star)^2.
\end{equation}
\end{proposition}

\begin{proof}
Fix $\latentData$ and choose $\mb Q\in\cO(d)$ so that $d_{\cG}(\latentData,\latentData^\star)=n^{-1/2}\norm{\latentData_c-\latentData_c^\star \mb Q}_F$. Write $\mb H=\latentData_c-\latentData_c^\star \mb Q$, and let $\mb h_i$ denote the $i$th row of $\mb H$. Then $\norm{\mb H}_F^2=n\,d_{\cG}(\latentData,\latentData^\star)^2$ and the rows of $\mb H$ sum to $0$.

For each fixed $y\in[0,M]$, the map $r\mapsto \ell_M(y,r)$ is $C^2$ on $[0,\sqrt d]$, and each of the three branches from \eqref{eq:cappedDensity} has second derivative uniformly bounded on the relevant compact interval. Hence there exists $H_\ell=H_\ell(d,\sigma,M)<\infty$ such that
\begin{equation}\label{app:eq:loss_second_derivative_bound}
\sup_{y\in[0,M]}\sup_{r\in[0,\sqrt d]}\big|\partial_r^2\ell_M(y,r)\big|\le H_\ell.
\end{equation}
By \eqref{app:eq:popRiskPt},
\begin{equation*}
g(r;r^\star)=\int \ell_M(y,r)q_{M,r^\star}(y)\,\nu(dy).
\end{equation*}
Using $\partial_r\ell_M(y,r)=-(\partial_r q_{M,r}(y))/q_{M,r}(y)$, we get
\begin{align*}
\partial_r g(r;r^\star)\big|_{r=r^\star}
&=-\int q_{M,r^\star}(y)\frac{\partial_r q_{M,r}(y)|_{r=r^\star}}{q_{M,r^\star}(y)}\,\nu(dy)\\
&=-\int \partial_r q_{M,r}(y)\big|_{r=r^\star}\,\nu(dy)\\
&=-\partial_r\left(\int q_{M,r}(y)\,\nu(dy)\right)\bigg|_{r=r^\star}=0,
\end{align*}
because $\int q_{M,r}(y)\,\nu(dy)=1$ for every $r$. Moreover, \eqref{app:eq:loss_second_derivative_bound} implies $|\partial_r^2 g(r;r^\star)|\le H_\ell$ for all $r,r^\star\in[0,\sqrt d]$. Taylor's theorem therefore gives
\begin{equation}\label{app:eq:g_upper_scalar}
g(r;r^\star)-g(r^\star;r^\star)
\le \frac{H_\ell}{2}(r-r^\star)^2,
\qquad r,r^\star\in[0,\sqrt d].
\end{equation}

Applying \eqref{app:eq:g_upper_scalar} with $r=r_{ij}(\latentData)$ and $r^\star=r_{ij}^\star$ and then using the fact that centering and rotation do not change pairwise distances, we obtain
\begin{equation}\label{app:eq:pair_upper_by_h}
g\big(r_{ij}(\latentData);r_{ij}^\star\big)-g\big(r_{ij}^\star;r_{ij}^\star\big)
\le \frac{H_\ell}{2}\big(r_{ij}(\latentData)-r_{ij}^\star\big)^2
\le \frac{H_\ell}{2}\norm{\mb h_i- \mb h_j}^2.
\end{equation}
Averaging \eqref{app:eq:pair_upper_by_h} over pairs and using the identity $\sum_{i<j}\norm{\mb h_i-\mb h_j}^2=n\norm{\mb H}_F^2$ yields
\begin{equation*}
R_n^0(\latentData)-R_n^0(\latentData^\star)
\le \frac{H_\ell}{2N}\sum_{i<j}\norm{\mb h_i-\mb h_j}^2
=\frac{H_\ell n\norm{\mb H}_F^2}{2N}
\le 2H_\ell\,d_{\cG}(\latentData,\latentData^\star)^2,
\end{equation*}
since $N=\binom n2\ge n^2/4$ for $n\ge2$. This proves \eqref{app:eq:global_upper_margin} with $C_{\rm up}=2H_\ell$.
\end{proof}

\subsection{Consistency of the exact posterior}\label{app:SecConsistencyPost}

\begin{proposition}[Consistency of the exact posterior]\label{app:prop:exact}
Suppose Assumption~\ref{app:A:model} holds. Then, for every $\delta>0$,
\begin{equation*}
\Pi_n\Big(\big\{\latentData:\ d_{\cG}(\latentData,\latentData^\star)>\delta\big\}\ \Big|\ \mb D\Big)\ \xrightarrow{\Pbb}\ 0.
\end{equation*}
\end{proposition}

\begin{proof}
Fix $\delta>0$, and define the  constant
\begin{equation}\label{app:eq:a0_def_exact}
a_0=\min\Big\{\frac12,\ \sqrt{\frac{c_{\rm mg}}{4C_{\rm up}}}\Big\}\in\Big(0,\frac12\Big],
\end{equation}
where $c_{\rm mg}$ is from Proposition~\ref{app:prop:marginOrbit} and $C_{\rm up}$ is from Proposition~\ref{app:prop:upperMarginOrbit}. Let
\begin{equation*}
A_\delta=\{\latentData:\ d_{\cG}(\latentData,\latentData^\star)>\delta\},
\qquad
B_{a_0\delta}=\{\latentData:\ d_{\cG}(\latentData,\latentData^\star)\le a_0\delta\},
\qquad
\widetilde B_{a_0\delta}=B_{a_0\delta}\cap([0,1]^d)^n.
\end{equation*}
Because the integrand is nonnegative, restricting the domain of integration gives
\begin{equation}\label{app:eq:Btild_def_support}
\int e^{-N R_n(\latentData)}\,\Pi_0(d\latentData)\ \ge\ \int_{\widetilde B_{a_0\delta}} e^{-N R_n(\latentData)}\,\Pi_0(d\latentData).
\end{equation}

Let $\eta_n=\sup_{\latentData\in([0,1]^d)^n}|R_n(\latentData)-R_n^0(\latentData)|$. By Lemma~\ref{app:lem:ULLN}, $\eta_n\to 0$ in probability. Let
\begin{equation*}
E_n^{\mathrm{cond}}=\{s_{\min}(\latentData_c^\star)\ge c_1\sqrt n\},
\end{equation*}
where $c_1=c_1(d)$ is the constant from Lemma~\ref{app:lem:conditioning}; then $\Pbb(E_n^{\mathrm{cond}})\to1$. Fix $\gamma>0$. On the event $\{\eta_n\le \gamma\}$ we have, for all $\latentData$,
\begin{equation}\label{app:eq:Rn_sandwich}
R_n(\latentData)\ge R_n^0(\latentData)-\gamma,\qquad R_n(\latentData)\le R_n^0(\latentData)+\gamma.
\end{equation}
Therefore, using \eqref{app:eq:Btild_def_support} and then \eqref{app:eq:Rn_sandwich},
\begin{align*}
\Pi_n(A_\delta\mid \mb D)
&=
\frac{\int_{A_\delta} e^{-N R_n(\latentData)}\Pi_0(d\latentData)}{\int e^{-N R_n(\latentData)}\Pi_0(d\latentData)}
\le
\frac{\int_{A_\delta} e^{-N R_n(\latentData)}\Pi_0(d\latentData)}{\int_{\widetilde B_{a_0\delta}} e^{-N R_n(\latentData)}\Pi_0(d\latentData)}\\
&\le
\frac{\int_{A_\delta} e^{-N (R_n^0(\latentData)-\gamma)}\Pi_0(d\latentData)}{\int_{\widetilde B_{a_0\delta}} e^{-N (R_n^0(\latentData)+\gamma)}\Pi_0(d\latentData)}
=
e^{2N\gamma}\cdot
\frac{\int_{A_\delta} e^{-N R_n^0(\latentData)}\Pi_0(d\latentData)}{\int_{\widetilde B_{a_0\delta}} e^{-N R_n^0(\latentData)}\Pi_0(d\latentData)}.
\end{align*}

Next, on $E_n^{\mathrm{cond}}$, Proposition~\ref{app:prop:marginOrbit} gives
\begin{equation}\label{app:eq:exact_num_margin}
R_n^0(\latentData)\ge R_n^0(\latentData^\star)+c_{\rm mg}\delta^2\qquad\text{for all }\latentData\in A_\delta.
\end{equation}
Also, Proposition~\ref{app:prop:upperMarginOrbit} and the choice of $a_0$ in \eqref{app:eq:a0_def_exact} imply
\begin{equation}\label{app:eq:exact_den_margin}
R_n^0(\latentData)\le R_n^0(\latentData^\star)+C_{\rm up}a_0^2\delta^2
\le R_n^0(\latentData^\star)+\frac14 c_{\rm mg}\delta^2
\qquad\text{for all }\latentData\in \widetilde B_{a_0\delta}.
\end{equation}
Using \eqref{app:eq:exact_num_margin} and \eqref{app:eq:exact_den_margin},
\begin{equation}\label{app:eq:num_bound_exact}
\int_{A_\delta} e^{-N R_n^0(\latentData)}\Pi_0(d\latentData)
\le
e^{-N(R_n^0(\latentData^\star)+c_{\rm mg}\delta^2)}\Pi_0([0,1]^{dn}),
\end{equation}
and
\begin{equation}\label{app:eq:den_bound_exact}
\int_{\widetilde B_{a_0\delta}} e^{-N R_n^0(\latentData)}\Pi_0(d\latentData)
\ge
e^{-N(R_n^0(\latentData^\star)+c_{\rm mg}\delta^2/4)}\Pi_0(\widetilde B_{a_0\delta}).
\end{equation}
Combining Equations~\ref{app:eq:num_bound_exact}--\ref{app:eq:den_bound_exact} on $E_n^{\mathrm{cond}}\cap\{\eta_n\le\gamma\}$ yields
\begin{equation}\label{app:eq:post_ratio_exact}
\Pi_n(A_\delta\mid \mb D)
\le
\frac{\Pi_0([0,1]^{dn})}{\Pi_0(\widetilde B_{a_0\delta})}
\exp\Big(-N\cdot \tfrac{3}{4}c_{\rm mg}\delta^2 + 2N\gamma\Big).
\end{equation}

It remains to lower-bound $\Pi_0(\widetilde B_{a_0\delta})$. By Assumption~\ref{app:A:model}(T3), if $\pi_0$ denotes the prior density then
\begin{equation}\label{app:eq:prior_density_lower}
\pi_0(\latentData)\ge e^{-cn},
\qquad \latentData\in([0,1]^d)^n,
\end{equation}
for some constant $c>0$. Moreover, by the definition
\begin{equation*}
d_{\cG}(\latentData,\latentData^\star)=\frac1{\sqrt n}\min_{\mb Q\in\cO(d)}\norm{\latentData_c-\latentData_c^\star \mb Q}_F,
\end{equation*}
we have the inclusion
\begin{equation}\label{app:eq:ball_inclusion_rms}
\Big\{\latentData\in([0,1]^d)^n:\ \norm{\latentData-\latentData^\star}_F\le \sqrt n\,a_0\delta\Big\}\ \subseteq\ \widetilde B_{a_0\delta}.
\end{equation}
Set
\begin{equation*}
b_\delta=\min\!\Big\{1,\frac{a_0\delta}{\sqrt d}\Big\},
\end{equation*}
and write $\|\cdot\|_{\max}$ for the maximum absolute coordinate over the $n\times d$ matrix. Since $\norm{\latentData-\latentData^\star}_{\max}\le a_0\delta/\sqrt d$ implies $\norm{\latentData-\latentData^\star}_{F}\le \sqrt n\,a_0\delta$, the box
\begin{equation*}
C_\delta=\Big\{\latentData\in([0,1]^d)^n:\ \norm{\latentData-\latentData^\star}_{\max}\le a_0\delta/\sqrt d\Big\}
\end{equation*}
is contained in the ball $\widetilde B_{a_0\delta}$. 

Note that each side of this box has length at least $b_\delta$, so $\mathrm{Vol}(C_\delta)\ge b_\delta^{nd}$. Therefore, by \eqref{app:eq:prior_density_lower},
\begin{equation}\label{app:eq:prior_mass_lower}
\Pi_0(\widetilde B_{a_0\delta})
\ \ge\ e^{-cn}\cdot \mathrm{Vol}(C_\delta)
\ \ge\ e^{-cn}\,b_\delta^{nd}.
\end{equation}
In particular,
\begin{equation}\label{app:eq:log_prior_mass_exact}
-\log \Pi_0(\widetilde B_{a_0\delta})\ \le\ cn + nd\log\!\Big(\frac{1}{b_\delta}\Big),
\end{equation}
which is $O(n)$ for each fixed $\delta>0$.

Fix $\tau=\frac{3}{16}c_{\rm mg}\delta^2$. Since $\eta_n\to 0$ in probability and $\Pbb(E_n^{\mathrm{cond}})\to1$, the event $\{\eta_n\le \tau\}\cap E_n^{\mathrm{cond}}$ has probability tending to $1$. On this event, \eqref{app:eq:post_ratio_exact} with $\gamma=\tau$ and \eqref{app:eq:log_prior_mass_exact} give
\begin{equation*}
\Pi_n(A_\delta\mid \mb D)
\le
\frac{\Pi_0([0,1]^{dn})}{\Pi_0(\widetilde B_{a_0\delta})}
\exp\Big(-N\cdot \tfrac{3}{8}c_{\rm mg}\delta^2\Big),
\end{equation*}
and the right-hand side tends to $0$ because $N=\binom{n}{2}\asymp n^2$ while \eqref{app:eq:log_prior_mass_exact} is $O(n)$ for fixed $\delta$. This proves the claim.
\end{proof}

\subsection{Proof of Proposition~\ref{app:main_prop}}

Let $\varepsilon_n$ be as in \eqref{app:eq:pairwiseError}, and define
\begin{equation*}
E_n^{\mathrm{app}}
=
\left\{
\sup_{\latentData\in([0,1]^d)^n}
|\widehat R_n(\latentData)-R_n(\latentData)|
\le \varepsilon_n
\right\}.
\end{equation*}
By \eqref{app:eq:riskError}, $\Pbb(E_n^{\mathrm{app}})\to1$. On $E_n^{\mathrm{app}}$, \eqref{app:eq:riskError} gives
\begin{equation*}
\sup_{\latentData\in([0,1]^d)^n}|\widehat R_n(\latentData)-R_n(\latentData)|\le \varepsilon_n.
\end{equation*}
Repeat the proof of Proposition~\ref{app:prop:exact} with $\widehat R_n$ in place of $R_n$.
Everywhere that the previously defined quantity $\eta_n$ appeared, we now have, on $E_n^{\mathrm{app}}$,
\begin{equation}\label{app:eq:hat_eta_bound}
\widehat \eta_n = \sup_{\latentData\in([0,1]^d)^n}|\widehat R_n(\latentData)-R_n^0(\latentData)|
\le \sup_{\latentData\in([0,1]^d)^n}|\widehat R_n(\latentData)-R_n(\latentData)| + \sup_{\latentData\in([0,1]^d)^n}|R_n(\latentData)-R_n^0(\latentData)|
\le \varepsilon_n + \eta_n.
\end{equation}
Since $\Pbb(E_n^{\mathrm{app}})\to1$, $\varepsilon_n\to 0$ by assumption, and $\eta_n\to 0$ in probability by Lemma~\ref{app:lem:ULLN}, \eqref{app:eq:hat_eta_bound} gives $\widehat\eta_n\to 0$ in probability. The remainder of the argument (risk gap $\Rightarrow$ posterior concentration) is unchanged.

\subsection{Proof of Theorem~\ref{thm:NoisyPC_ext}} \label{SecFinalAppA}

Let
\begin{equation*}
\Delta_n^{\mathrm{mix}}
=
\sup_{\latentData_0\in([0,1]^d)^n}\big\|P(\latentData_0,\cdot)-P_{\xi_0}(\latentData_0,\cdot)\big\|_{\mathrm{TV}}.
\end{equation*}
By convexity of total variation and the definition of $P$,
\begin{equation}\label{app:eq:mix_delta_bound}
\Delta_n^{\mathrm{mix}}
\le \int_{\Xi}\sup_{\latentData_0\in([0,1]^d)^n}\big\|P_{\xi}(\latentData_0,\cdot)-P_{\xi_0}(\latentData_0,\cdot)\big\|_{\mathrm{TV}}\,\mu(d\xi)
\le \Delta_n.
\end{equation}
By the uniform-ergodic perturbation bound in \cite{AlquierFrielEverittBoland2016} and \eqref{app:eq:mix_delta_bound},
\begin{equation}\label{app:eq:kernel_perturb_bound}
\sup_{\latentData_0\in([0,1]^d)^n}\big\|\delta_{\latentData_0}P^m-\delta_{\latentData_0}P_{\xi_0}^m\big\|_{\mathrm{TV}}
\le \kappa_n\Delta_n^{\mathrm{mix}}
\le \kappa_n\Delta_n,
\qquad m\ge 0.
\end{equation}
Integrating \eqref{app:eq:kernel_perturb_bound} against $\latentData_0\sim \widetilde\Pi_n$ and using the invariance of $\widetilde\Pi_n$ under $P$ gives
\begin{equation}\label{app:eq:stationary_to_reference_m}
\big\|\widetilde\Pi_n-\widetilde\Pi_n P_{\xi_0}^m\big\|_{\mathrm{TV}}
\le \kappa_n\Delta_n.
\end{equation}
By the triangle inequality and \eqref{app:eq:stationary_to_reference_m},
\begin{align*}
\big\|\widetilde\Pi_n-\widehat\Pi_n^{(\xi_0)}\big\|_{\mathrm{TV}}
&\le \big\|\widetilde\Pi_n-\widetilde\Pi_n P_{\xi_0}^m\big\|_{\mathrm{TV}}
 + \big\|\widetilde\Pi_n P_{\xi_0}^m-\widehat\Pi_n^{(\xi_0)}\big\|_{\mathrm{TV}}.
\end{align*}
Since $P_{\xi_0}$ has invariant distribution $\widehat\Pi_n^{(\xi_0)}$ and is uniformly ergodic, the second term tends to $0$ as $m\to\infty$. Letting $m\to\infty$ yields
\begin{equation}\label{app:eq:invariant_tv_bound}
\big\|\widetilde\Pi_n-\widehat\Pi_n^{(\xi_0)}\big\|_{\mathrm{TV}}
\le \kappa_n\Delta_n.
\end{equation}
Now set $A_\delta=\{\latentData:\ d_{\cG}(\latentData,\latentData^\star)>\delta\}$. Then \eqref{app:eq:invariant_tv_bound} gives
\begin{equation*}
\widetilde\Pi_n(A_\delta)
\le \widehat\Pi_n^{(\xi_0)}(A_\delta)+\kappa_n\Delta_n.
\end{equation*}
The first term goes to $0$ in probability by Proposition~\ref{app:main_prop}, and the second does as well by assumption.

\clearpage
\section{Supplementary algorithms}\label{app:supalg}
\setcounter{algocf}{0}
\counterwithin{algocf}{section}

% alg - building quadtree
\begin{algorithm}[H]
\footnotesize
\caption{Quadtree construction}\label{alg:QT}
\KwIn{$\latentdata_1,\dots,\latentdata_n \in \RR^2$, $\highdata_1,\dots,\highdata_n \in \RR^\Omega$}
\KwOut{quadtree $\mathcal{T}$}
\SetKwFunction{Insert}{Insert}
\SetKwProg{Fn}{Function}{}{}
\Fn{\Insert{$i, k$}}{
    % Check boundary
    \If{$\latentdata_i \notin \mathcal{C}_k$}{
        \Return false
    }
    % Update summaries
    Update $N_k, \bar{\latentdata}_k, \bar{\highdata}_k$\;
    % If empty leaf add point
    \If{$\text{node }k$ is an empty leaf}{
        Store $(\latentdata_i, \highdata_i)$ in $node$ $k$\;
        \Return true
    }
    % Subdivide if leaf and not empty
    \If{$\text{node }k$ is a leaf}{
        Subdivide $\mathcal{C}_k$ into children $c \in children(k)$\;
        \For{each stored point $j$ in node $k$}{
            Determine $c(j)$ and call \Insert{$j, c(j)$}\;
        }
        Mark $node$ $k$ as internal\;
    }
    % Recurse
    \ForEach{$c \in \text{children}(k)$}{
        \If{\Insert{$i, c$}}{
            \Return true
        }
    }
    \Return false
}

Initialize root node $k=0$ with cell $\mathcal{C}_0$\;
\For{$i = 1,\dots,n$}{
    \Insert{$i, 0$}\;
}
\Return{$\mathcal{T}$}
\end{algorithm}

% alg -local quadtree updates
\begin{algorithm}[H]
\footnotesize
\caption{Local quadtree update (forward)}\label{alg:QT_updates}
\KwIn{quadtree $\mathcal{T}$, current point $\latentdata_i$, proposed point $\latentdata_i'$}
\KwOut{locally updated quadtree $\mathcal{T}$}
Let $\mathcal{I}_i = \{\text{node } k \in \mathcal{T} :  \latentdata_i \in \mathcal{C}_k\}$ be the set of nodes along the path from root to leaf containing $\latentdata_i$\;
\For{node $k \in \mathcal{I}_i$}{
    Cache $\bar{\latentdata}_k$\; 
    Update $\bar{\latentdata}_k \leftarrow \bar{\latentdata}_k + \frac{(\latentdata_i' - \latentdata_i)}{N_k}$
}
\Return{$\mathcal{T}$}
\end{algorithm}

% alg  - deterministic evaluation of the BH-BMDS likelihood (BH-BMDS)
\begin{figure}[h]
\centering
\vspace{-0.5em}
\begin{minipage}[t]{0.48\linewidth}
\vspace{0pt}
\begin{algorithm}[H]
\footnotesize
\caption{Deterministic evaluation of the BH-BMDS log-likelihood}\label{alg:deter}
\KwIn{$\latentdata_1,\dots, \latentdata_n \in \RR^2$; $\highdata_1, \dots, \highdata_n \in \RR^\Omega$; quadtree $\mathcal{T}$}
\KwOut{BH-BMDS log-likelihood $\tilde{\ell}(\cdot \mid \mathcal{T})$}
\SetKwFunction{BH}{BH}
\SetKwProg{Fn}{Function}{}{}
\Fn{\BH{$i, k$}}{
    Compute $\rho_{ik} = w_k / \bar r_{ik}$\;
    \eIf{node $k$ is leaf \textbf{or} $\rho_{ik} < \theta$}{
        \Return{
        $N_k \biggl[
            \frac{(\bar{D}_{ik} - \bar{r}_{ik})^2}{2\sigma^2}
            + \log \Phi\!\left( \frac{\bar{r}_{ik}}{\sigma} \right)
        \biggr]$
        }
    }{
        \Return{$\sum_{c \in \text{children}(k)} \BH(i, c)$}
    }
}
Initialize $\tilde{\ell}(\cdot \mid \mathcal{T}) = -N\log(2\pi\sigma^2)$\;
\For{$i = 1,\dots,n$}{
    $\tilde{\ell}(\latentdata_i \mid \mathcal{T}) \leftarrow \BH(i, 0)$\;
    $\tilde{\ell}(\cdot \mid \mathcal{T}) \leftarrow \tilde{\ell}(\cdot \mid \mathcal{T}) + \tilde{\ell}(\latentdata_i \mid \mathcal{T})$\;
}
\Return{$\tilde{\ell}(\cdot \mid \mathcal{T})$}
\end{algorithm}
\vspace{-0.5em}
\end{minipage}
\hfill
%alg 3 - stochastic pointwise evaluation of the BH-BMDS likelihood
\begin{minipage}[t]{0.48\linewidth}
\vspace{0pt}
\begin{algorithm}[H]
\footnotesize
\caption{Stochastic evaluation of the pointwise contribution to the BH-BMDS log-likelihood}\label{alg:stoch}
\KwIn{target point $\latentdata_i$, quadtree $\mathcal{T}$}
\KwOut{$\tilde{\ell}(\latentdata_i \mid \mathcal{T})$}
\SetKwFunction{SBH}{SBH}
\SetKwProg{Fn}{Function}{}{}
\Fn{\SBH{$i, k$}}{
    Draw $u \sim \text{Unif}(0,1)$\;
    Compute $p(\rho_{ik} \mid \theta)$ \Eqref{eq::logfunc}\;
    \eIf{node $k$ is leaf \textbf{or} $u < p(\rho_{ik} \mid \theta)$}{
        \Return{
        $N_k \biggl[
            \frac{(\bar{D}_{ik} - \bar{r}_{ik})^2}{2\sigma^2}
            + \log \Phi\!\left( \frac{\bar{r}_{ik}}{\sigma} \right)
        \biggr]$
        }
    }{
        \Return{$\sum_{c \in \text{children}(k)} \SBH(i, c)$}
    }
}
$\tilde{\ell}(\latentdata_i \mid \mathcal{T}) \leftarrow \texttt{SBH(i, 0)}$\;
\Return{$\tilde{\ell}(\latentdata_i \mid \mathcal{T})$}
\end{algorithm}
\vspace{-0.5em}
\end{minipage}
\vspace{-1em}
\end{figure}

\clearpage
\section{Additional plots}
\setcounter{figure}{0}
\counterwithin{figure}{section}

\begin{figure}[H]
    \centering
    \includegraphics[width=0.9\linewidth]{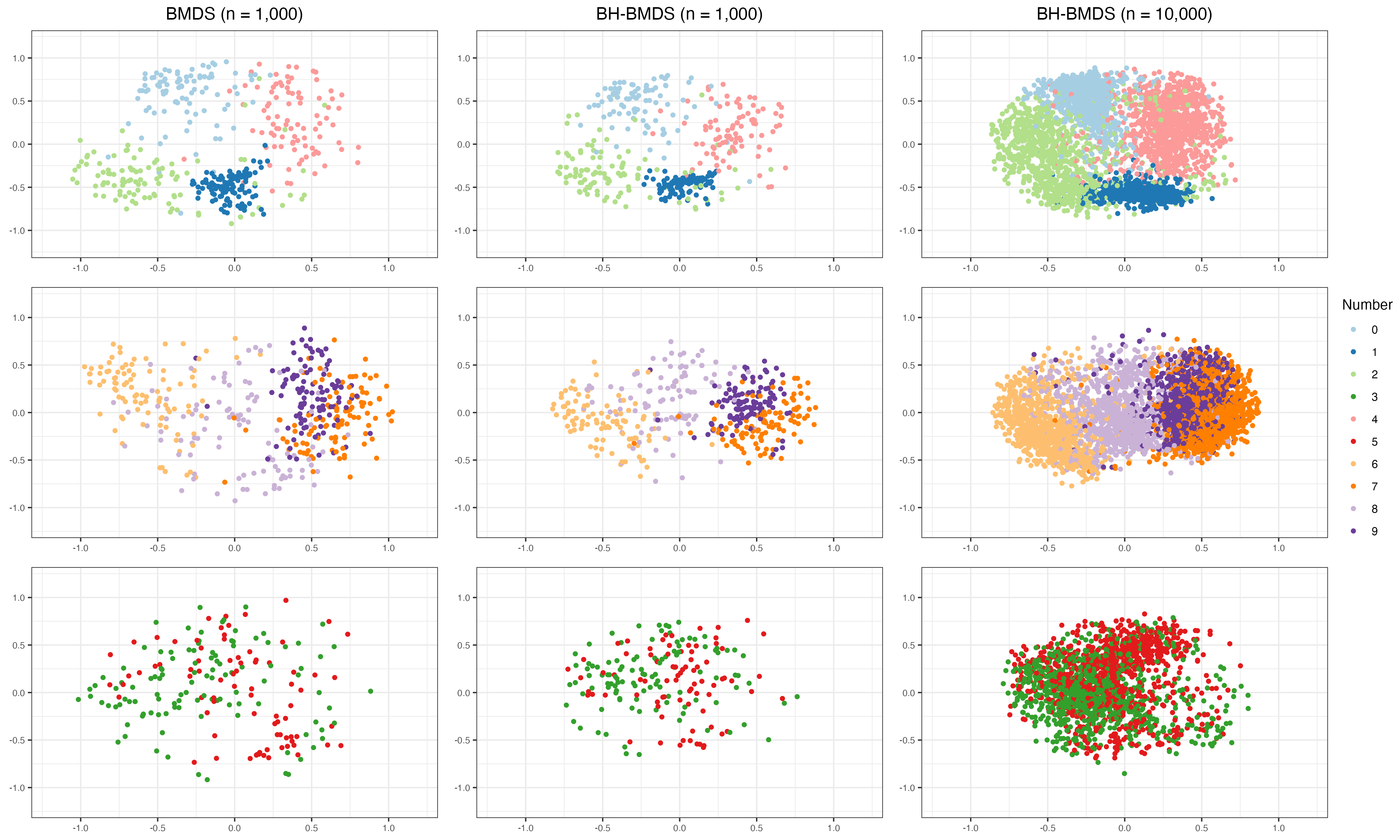}
    \caption{Posterior Procrustes-aligned means across MCMC iterations for $1{,}000$ and $10{,}000$ embedded MNIST handwritten digital images under full Bayesian multidimensional scaling (BMDS) and Barnes--Hut Bayesian multidimensional scaling (BH-BMDS) frameworks.}
    \label{fig:mnist_all}
\end{figure}

\begin{figure}[H]
    \centering
    \includegraphics[width=0.9\linewidth]{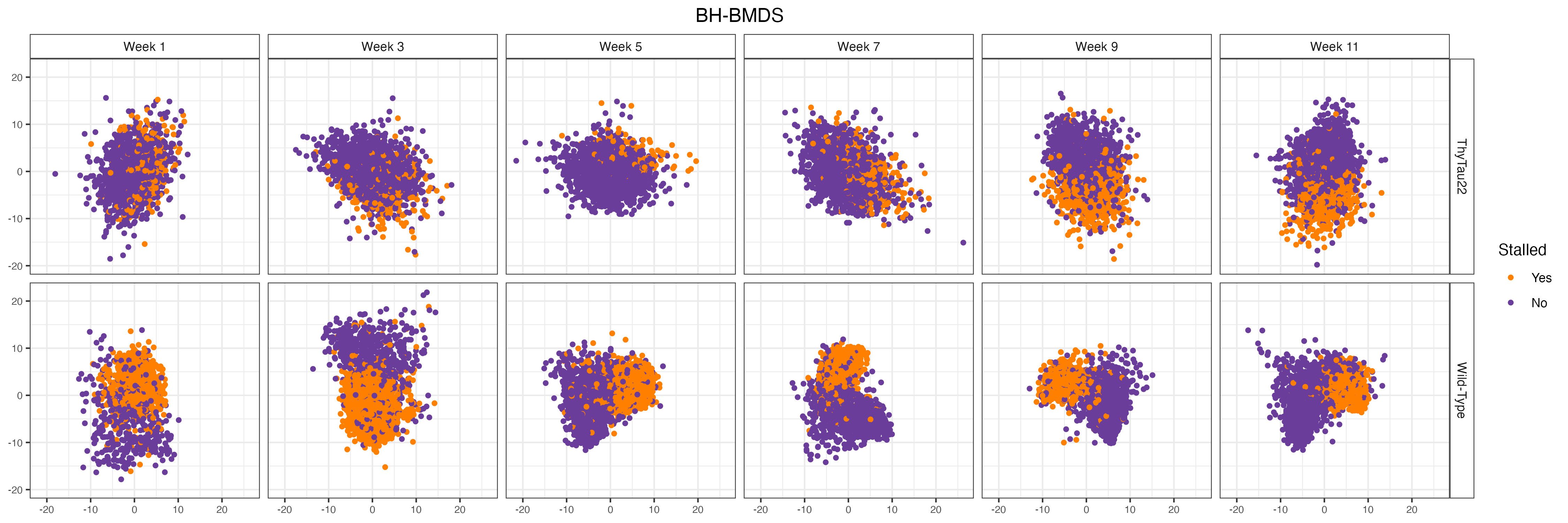}
    \caption{Posterior Procrustes-aligned means across MCMC iterations for ThyTau22 and Wild-Type mouse under the Barnes--Hut Bayesian multidimensional scaling (BH-BMDS) framework. High-dimensional embeddings are obtain from a trained MaGNet GNN \citep{magnet} to decode behavior states (stalled vs.~moving) on local field potential recordings. We observe more clear clustering for a Wild-Type than ThyTau22 mouse.}
    \label{fig:alz_all}
\end{figure}
\end{document}